\documentclass[12pt, onecolumn,a4paper]{quantumarticle}
\pdfoutput=1 

\usepackage{thmtools}
\usepackage{thm-restate}
\usepackage{makecell}
\usepackage{tikz}
\usetikzlibrary{positioning, calc}
\usetikzlibrary{calc}
\usetikzlibrary{arrows}
\usepackage{tikz-3dplot}
\usepackage{graphicx}
\usepackage{epsfig}
\usepackage{mdwlist}
\usepackage{color}
\usepackage{thmtools}
\usepackage{amssymb}
\usepackage{physics}
\usepackage{amsmath}
\usepackage{array}
\usepackage{doi}
\usepackage{url,hyperref}
\hypersetup{colorlinks={true},linkcolor=.,citecolor=red}

\usepackage[numbers,sort&compress]{natbib}
\usetikzlibrary{fadings}
\usetikzlibrary{decorations.pathmorphing}
\usepackage{mathrsfs}
\usepackage{authblk}
\usepackage[mathscr]{euscript}
\usepackage{enumitem}
\usetikzlibrary{quantikz2}

\DeclareMathAlphabet\mathbfcal{OMS}{cmsy}{b}{n}

\usepackage[noabbrev]{cleveref}

\usetikzlibrary{decorations.pathreplacing}
\tikzset{snake it/.style={decorate, decoration=snake}}

\usetikzlibrary{decorations.pathreplacing,decorations.markings}

\tikzset{
    >=stealth',
    punkt/.style={
           rectangle,
           rounded corners,
           draw=black, very thick,
           text width=6.5em,
           minimum height=2em,
           text centered},
    pil/.style={
           ->,
           thick,
           shorten <=2pt,
           shorten >=2pt,},
  on each segment/.style={
    decorate,
    decoration={
      show path construction,
      moveto code={},
      lineto code={
        \path [#1]
        (\tikzinputsegmentfirst) -- (\tikzinputsegmentlast);
      },
      curveto code={
        \path [#1] (\tikzinputsegmentfirst)
        .. controls
        (\tikzinputsegmentsupporta) and (\tikzinputsegmentsupportb)
        ..
        (\tikzinputsegmentlast);
      },
      closepath code={
        \path [#1]
        (\tikzinputsegmentfirst) -- (\tikzinputsegmentlast);
      },
    },
  },
  mid arrow/.style={postaction={decorate,decoration={
        markings,
        mark=at position .5 with {\arrow[#1]{stealth'}}
      }}}
}

\usepackage{slashed}
\usepackage{float} % this is to place figures where requested!
\usepackage{graphicx} % needed for the figures
\usepackage{subcaption}

\mathchardef\mhyphen="2D

\newcommand{\CDQS}{\textnormal{CDQS}}

\newtheorem{theorem}{Theorem}

\newtheorem{definition}[theorem]{Definition}

\newtheorem{lemma}[theorem]{Lemma}

\newtheorem{remark}[theorem]{Remark}

\newenvironment{proof}[1][Proof]{\noindent\textbf{#1. }}{\ \rule{0.5em}{0.5em}}
\usepackage[T1]{fontenc}

\begin{document}

\title{A complexity theory for non-local quantum computation}

\author[1]{Andreas Bluhm}
\email{andreas.bluhm@univ-grenoble-alpes.fr}
\orcid{0000-0003-4796-7633}

\author[1]{Simon H\"{o}fer}
\email{simon.hofer@univ-grenoble-alpes.fr}
\orcid{}

\author[2,3]{Alex May}
\email{amay@perimeterinstitute.ca}
\orcid{0000-0002-4030-5410}

\author[2]{Mikka Stasiuk}
\email{mstasiuk@perimeterinstitute.ca }
\orcid{}

\author[4]{Philip Verduyn Lunel}
\email{philip.verduyn-lunel@lip6.fr}
\orcid{0000-0001-5419-6027}

\author[5]{Henry Yuen}
\email{hyuen@cs.columbia.edu}
\orcid{0000-0002-2684-1129}

\affiliation[1]{Univ. Grenoble Alpes, CNRS, Grenoble INP, LIG}
\affiliation[2]{Perimeter Institute for Theoretical Physics}
\affiliation[3]{Institute for Quantum Computing, Waterloo, Ontario}
\affiliation[4]{Sorbonne Universit\'e, Paris}
\affiliation[5]{Columbia University}

\abstract{
Non-local quantum computation (NLQC) replaces a local interaction between two systems with a single round of communication and shared entanglement. 
Despite many partial results, it is known that a characterization of entanglement cost in at least certain NLQC tasks would imply significant breakthroughs in complexity theory. 
Here, we avoid these obstructions and take an indirect approach to understanding resource requirements in NLQC, which mimics the approach used by complexity theorists: we study the relative hardness of different NLQC tasks by identifying resource efficient reductions between them. 
Most significantly, we prove that $f$-measure and $f$-route, the two best studied NLQC tasks, are in fact equivalent under $O(1)$ overhead reductions.
This result simplifies many existing proofs in the literature and extends several new properties to $f$-measure. 
For instance, we obtain sub-exponential upper bounds on $f$-measure for all functions, and efficient protocols for functions in the complexity class $\mathsf{Mod}_k\mathsf{L}$. 
Beyond this, we study coherently controlled applications of Pauli operators and more general unitaries. 
We find relationships among coherent protocols, and show the coherently controlled NLQCs imply classically controlled protocols, suggesting they are harder NLQC tasks.
}

\maketitle

\pagebreak

\tableofcontents

\flushbottom

\section{Introduction} 

A non-local quantum computation (NLQC) replaces a local interaction of two systems with a single round of communication and shared entanglement (see \Cref{fig:non-localandlocal}).
Non-local quantum computation has many applications, including to quantum position-verification (QPV) \cite{kent2011quantum,buhrman2014position}, AdS/CFT \cite{may2019quantum,may2020holographic,may2022complexity,kubicki2024constraints}, information-theoretic cryptography \cite{allerstorfer2024relating,asadi2024conditional,girish2025comparing}, uncloneable secret sharing \cite{ananth2025unclonable}, and Hamiltonian complexity \cite{apel2024security}. 
In all these applications, a key quantity of interest is the amount of entanglement necessary to implement a given computation as a NLQC. 
Despite the need to understand this, a clear and general characterization of the entanglement cost of a NLQC is still lacking. 

\begin{figure*}
    \centering
    \begin{subfigure}{0.45\textwidth}
    \centering
    \begin{tikzpicture}[scale=0.6]
    
    %interaction unitary unitary
    \draw[thick] (-1,-1) -- (-1,1) -- (1,1) -- (1,-1) -- (-1,-1);
    
    %wire into interaction unitary
    \draw[thick,mid arrow] (-3.5,-3) to [out=90,in=-90] (-0.5,-1);
    \draw[thick,mid arrow] (3.5,-3) to [out=90,in=-90] (0.5,-1);
    
    \draw[thick,mid arrow] (0.5,1) to [out=90,in=-90] (3.5,3);
    \draw[thick,mid arrow] (-0.5,1) to [out=90,in=-90] (-3.5,3);
    
    \node at (0,0) {$\mathbfcal{N}$};
    
    \node at (0,-5) {$ $};
    
    \end{tikzpicture}
    \caption{}
    \label{fig:local}
    \end{subfigure}
    \hfill
    \begin{subfigure}{0.45\textwidth}
    \centering
    \begin{tikzpicture}[scale=0.4]
    
    %lower left box
    \draw[thick] (-5,-5) -- (-5,-3) -- (-3,-3) -- (-3,-5) -- (-5,-5);
    \node at (-4,-4) {$\mathbfcal{V}^L$};
    
    %lower right box
    \draw[thick] (5,-5) -- (5,-3) -- (3,-3) -- (3,-5) -- (5,-5);
    \node at (4,-4) {$\mathbfcal{V}^R$};
    
    %top right box
    \draw[thick] (5,5) -- (5,3) -- (3,3) -- (3,5) -- (5,5);
    \node at (4,4) {$\mathbfcal{W}^R$};
    
    %top left box
    \draw[thick] (-5,5) -- (-5,3) -- (-3,3) -- (-3,5) -- (-5,5);
    \node at (-4,4) {$\mathbfcal{W}^L$};
    
    %left vertical wire
    \draw[thick,mid arrow] (-4.5,-3) -- (-4.5,3);
    
    %right vertical wire
    \draw[thick,mid arrow] (4.5,-3) -- (4.5,3);
    
    %left to right wire
    \draw[thick,mid arrow] (-3.5,-3) to [out=90,in=-90] (3.5,3);
    
    %right to left wire
    \draw[thick,mid arrow] (3.5,-3) to [out=90,in=-90] (-3.5,3);
    
    %entanglement
    \draw[thick] (-3.5,-6) -- (3.5,-6) -- (0,-8) -- (-3.5,-6);
    \draw[thick] (-3.25,-6) -- (-3.25,-5);
    \draw[thick] (3.25,-6) -- (3.25,-5);
    \node at (0,-7) {$\Psi$};
    
    %input wires
    \draw[thick] (-4.5,-6) -- (-4.5,-5);
    \draw[thick] (4.5,-6) -- (4.5,-5);
    
    %output wires
    \draw[thick] (4.5,5) -- (4.5,6);
    \draw[thick] (-4.5,5) -- (-4.5,6);
    
    \end{tikzpicture}
    \caption{}
    \label{fig:non-localcomputation}
    \end{subfigure}
    \caption{(a) Circuit diagram showing the local implementation of a channel $\mathbfcal{N}$. Physically, systems $A$ and $B$ are brought together spatially and interacted. (b) Circuit diagram showing the form of a non-local quantum computation. $\mathbfcal{V}^L$, $\mathbfcal{V}^R$, $\mathbfcal{W}^L$, and $\mathbfcal{W}^R$ are quantum channels. The triangle labelled $\Psi$ represents the shared, in general entangled, resource state. The goal is to simulate the local channel $\mathbfcal{N}$.}
    \label{fig:non-localandlocal}
\end{figure*}

There are fundamental obstructions to obtaining a complete understanding of entanglement cost in NLQC. 
In particular, entanglement cost in certain NLQC tasks can be upper bounded by complexity measures. 
Consequently, placing lower bounds on entanglement would imply lower bounds on these measures of complexity, which is a notoriously hard problem. 
For instance, in the context of a specific scheme known as $f$-route, there is an upper bound on entanglement cost which is exponential in the memory needed to compute the Boolean function $f$ \cite{buhrman2013garden}. 
Super-polynomial bounds on $f$-route for a given function $f$ therefore imply super-logarithmic lower bounds on memory. 
Thus such lower bounds face complexity barriers: a super-polynomial lower bound for a function in $\mathsf{P}$ (poly-time) would separate $\mathsf{L}$ (logspace) and $\mathsf{P}$, a long-standing problem in complexity theory.\footnote{In fact, somewhat more is true: in \cite{cree2023code} an upper bound which is efficient for functions in $Mod_kL$ is given, so that a super-polynomial lower bound on $f$-route for a function in $\mathsf{P}$ would separate $Mod_kL$, which contains $L$, and $P$.} 

In this work we take a new approach to elucidating how entanglement cost is determined in non-local quantum computation. 
We focus on the question: \emph{when is one computation harder to implement non-locally than another?}
To understand this, we study the relationships among different NLQC tasks, looking for reductions among them. 
Whenever we find that a protocol for implementing a non-local computation $A$ can be transformed, without too much overhead, into a protocol for implementing a computation $B$, we know that $B$ is no harder than $A$. 
By introducing this notion of a reduction we lay the groundwork for a more complete notion of complexity theory for NLQC.

Our strategy mimics the study of reductions among computational problems. 
Indeed, since we face some of the same obstructions that appear in the study of complexity theory, it is natural to adopt a similar strategy.
We can also hope to gain many of the same insights.
For instance, studying examples and comparing their hardness gives insight into what makes a problem hard or easy as a NLQC. 

We can also then try to identify the hardest NLQCs, which provide useful candidates for position-based cryptography.
Further, by building a scaffolding of examples and relating them to one another, any new NLQC of interest can be related to these problems and some of its properties understood without studying each new example from scratch.
Finally, we can hope to gain insight into complexity theory itself via NLQC: since upper bounds for different NLQC tasks are characterized in terms of different complexity measures, we can hope to learn about these complexity measures by studying reductions among the tasks, or even hope to find lower bounds on complexity measures by lower bounding NLQC tasks. 

Among the many applications of NLQC, our results are most immediately relevant for the understanding of the security of quantum position-verification.
We briefly review the relevance of NLQC to QPV here. 
In a position-verification scheme, the verifier sends the prover quantum and classical systems and asks for a reply at a set of designated spacetime locations.
See \Cref{fig:2dsetup} for a standard set-up in a spacetime with one spatial dimension. 
Unfortunately, there is no unconditionally secure QPV scheme; the cheating strategy is exactly to implement a non-local quantum computation in order to simulate the actions of an honest player. 
Since all computations can be implemented in the NLQC form, it follows that QPV is not secure without making further assumptions. 

\begin{figure}
    \centering
    \begin{subfigure}{0.45\textwidth}
    \begin{tikzpicture}[scale=0.7]
    
    \node[below left] at (-4,0) {$c_1$};
    \draw[fill=black] (-4,0) circle (0.15);

    \node[below right] at (4,0) {$c_2$};
    \draw[fill=black] (4,0) circle (0.15);

    \node[below right] at (4,8) {$r_2$};
    \draw[fill=blue] (4,8) circle (0.15);

    \node[below left] at (-4,8) {$r_1$};
    \draw[fill=blue] (-4,8) circle (0.15);
    
    \draw[fill=gray,opacity=0.5] (-1,1) -- (1,1) -- (1,7) -- (-1,7) -- (-1,1);

    \draw[->] (-5.5,1) -- (-5.5,2);
    \node[above] at (-5.5,2) {$t$};
    \draw[->] (-5.5,1) -- (-4.5,1);
    \node[right] at (-4.5,1) {$x$};
    
    \end{tikzpicture}
    \caption{}
    \label{fig:taggingsub1}
    \end{subfigure}
    \hfill
\begin{subfigure}{.45\textwidth}
\begin{tikzpicture}[scale=0.7]

    \node[below left] at (-4,0) {$c_1$};
    \draw[fill=black] (-4,0) circle (0.15);

    \node[below right] at (4,0) {$c_2$};
    \draw[fill=black] (4,0) circle (0.15);

    \node[below right] at (4,8) {$r_2$};
    \draw[fill=blue] (4,8) circle (0.15);

    \node[below left] at (-4,8) {$r_1$};
    \draw[fill=blue] (-4,8) circle (0.15);
    
    \draw[postaction={on each segment={mid arrow}}] (-4,0) -- (-2,2) -- (-2,6) -- (-4,8);
    \draw[postaction={on each segment={mid arrow}}] (4,0) -- (2,2) -- (2,6) -- (4,8);
    \draw[postaction={on each segment={mid arrow}}] (-2,2) -- (0,4) -- (2,6);
    \draw[postaction={on each segment={mid arrow}}] (2,2) -- (0,4) -- (-2,6);
    
    \draw[dashed] (2,2) -- (0,0) -- (-2,2);
    \node[below] at (0,0) {$\ket{\Psi}$};
    
    \draw[fill=yellow] (-2,2) circle (0.3);
    \draw[fill=yellow] (2,2) circle (0.3);
    \draw[fill=yellow] (-2,6) circle (0.3);
    \draw[fill=yellow] (2,6) circle (0.3);
    
    \draw[fill=gray,opacity=0.5] (-1,1) -- (1,1) -- (1,7) -- (-1,7) -- (-1,1);
    
\end{tikzpicture}
\caption{}
\label{fig:taggingsub2}
\end{subfigure}
    
\caption{A standard QPV set-up in $1+1$ dimensions. Inputs are given at spacetime locations $c_1$, $c_2$. The prover should apply a designated quantum operation to these inputs, then return the outputs to points $r_1$, $r_2$. a) An honest prover enters the designated spacetime region (grey) to apply the needed quantum operation. b) A dishonest prover attempts to reproduce the same operation while acting outside the spacetime region; their actions take the form of a non-local quantum computation. Figure reproduced from \cite{may2019quantum}.}
\label{fig:2dsetup}
\end{figure}

Rather than look for unconditional security, we can ask if some computations are much easier to implement in the local, honest, form rather than the non-local, dishonest form. 
Attention has focused on the case where we limit the entanglement shared by Alice and Bob, where the goal is to show that some computation which is reasonably easy to implement locally requires large entanglement to implement non-locally. See for example \cite{beigi2011simplified,tomamichel2013monogamy,bluhm2021position, may2022complexity, gonzales2019bounds, asadi2025linear, asadi2024rank}.

In this bounded entanglement setting, particular attention has been paid to classes of protocols where most of the input is classical, with just $O(1)$ quantum bits, and in particular to schemes where an honest prover need only compute a classical function and do $O(1)$ quantum operations. 
The two natural protocols of this form are $f$-routing and $f$-measure\footnote{This is also referred to as $f$-BB84 in the literature.}, which are leading candidates for practical schemes.  
In $f$-route, Alice receives a classical string $x\in \{0,1\}^n$ along with a quantum system $Q$ described by Hilbert space $\mathcal{H}_Q$, while Bob receives a classical string $y\in\{0,1\}^n$. 
$Q$ is in an arbitrary unknown state.
Their goal is to bring $Q$ to Alice's side if $f(x,y)=0$, or $Q$ to Bob's side if $f(x,y)=1$. 
In $f$-measure, again Alice receives a classical string $x\in \{0,1\}^n$ along with a quantum system $Q$, while Bob receives a classical string $y\in\{0,1\}^n$. 
Now however, $Q$ is one of the BB84 states $H^{f(x,y)}\ket{b}$, $q,b\in\{0,1\}$. 
Alice and Bob's goal is to both output the bit $b$ at the end of the NLQC. 

%%%%%%%%%%%%%%%%%%%%%%%%%%%%%%%%%%%%%%%%%%%%%%%%%%%%%%%%%%%%
\subsection{Related work}
%%%%%%%%%%%%%%%%%%%%%%%%%%%%%%%%%%%%%%%%%%%%%%%%%%%%%%%%%%%%

Our explorations are in part inspired by a related set of reductions between instances of NLQC and primitives in information-theoretic cryptography \cite{allerstorfer2024relating}. 
Those relationships have already led to a number of useful results \cite{asadi2024conditional, asadi2024rank}, and our work can be seen as extending this to the study of relationships between NLQC tasks themselves. 
We describe the results of \cite{allerstorfer2024relating} briefly below.

In \cite{allerstorfer2024relating}, the $f$-route task was related to a primitive studied in information-theoretic cryptography known as conditional disclosure of secrets (CDS) \cite{gertner1998protecting}.
In CDS two players, Alice and Bob, share correlations (in the classical setting) or entanglement (in the quantum setting) but do not communicate. 
Alice receives a string $s$, the secret, as well as an input string $x\in \{0,1\}^n$; Bob receives an input string $y\in \{0,1\}^n$. 
Alice and Bob can each send a message to a third party, the referee. 
The referee knows $x,y$ and wishes to learn $s$. 
Alice and Bob's goal is to allow the referee to learn $s$ if and only if $f(x,y)=1$ for some Boolean function $f$. 
The work \cite{allerstorfer2024relating} proved that CDS with quantum resources is equivalent to $f$-route: a protocol for quantum CDS can be transformed into one for $f$-route with only a small overhead in the resources used, and vice versa. 
Further, classical CDS schemes give quantum CDS schemes using the same resources.

A second result in \cite{allerstorfer2024relating} related another primitive from information-theoretic cryptography and a second instance of NLQC.
Private-simultaneous message passing (PSM) \cite{ishai1997private} is defined by a choice of Boolean function $f:\{0,1\}^{2n}\rightarrow \{0,1\}$ played by three parties, Alice, Bob and the referee. 
Alice receives input $x\in \{0,1\}^n$; Bob receives input $y\in \{0,1\}^n$.
In this setting the referee does not know $x,y$.
Alice and Bob will each send a message to the referee, who should learn $f(x,y)$ but no other information about $(x,y)$. 
PSM with quantum resources, denoted PSQM, is related to an instance of NLQC known as coherent function evaluation (CFE). 
In CFE, the goal is to implement the isometry
\begin{align}\label{eq:CFE}
    \mathbf{V}[f] = \sum_{x,y} \ket{x,y}_Z \ket{f(x,y)}_{Z'} \bra{x}_X \bra{y}_Y
\end{align}
with $X$ starting on Alice's side, $Y$ starting on Bob's side, and $Z$ ending on Alice's side, $Z'$ ending on Bob's side. 
In \cite{allerstorfer2024relating} a reduction transforming protocols for CFE into protocols for PSQM using similar resources was given. 

%%%%%%%%%%%%%%%%%%%%%%%%%%%%%%%%%%%%%%%%%%%%%%%%%%%%%%%%%%%%
\subsection{Our results}
%%%%%%%%%%%%%%%%%%%%%%%%%%%%%%%%%%%%%%%%%%%%%%%%%%%%%%%%%%%%

\begin{figure}
    \centering
    \begin{tikzpicture}[scale=0.75]
    \coordinate (fHbasis) at (5, 3);
    \coordinate (fR) at (0, 3);
    \coordinate (Cphase) at (5, 6);
    \coordinate (CSWAP) at (0,6);
    \coordinate (CFE) at (20,3);
    \coordinate (PSQM) at (15,3);
    \coordinate (CDS) at (10,0);
    \coordinate (PSM) at (15,0);
    \coordinate (MIP) at (-5,-6);
    \coordinate (CDQS) at (10,3);
    \coordinate (CPauli) at (10,6);
    
    \draw[thick,mid arrow, blue]  (CSWAP) --  (Cphase) node[midway, inner sep=0pt, outer sep=0pt] {\hyperref[lem:cfswapTOcfphase]{\phantom{\rule{1.3cm}{0.2cm}}}};
    \draw[thick,mid arrow, blue]  (CSWAP) --  (fR) node[midway, inner sep=0pt, outer sep=0pt] {\hyperref[remark:coherenttoincoherent]{\phantom{\rule{0.2cm}{1.2cm}}}};
    \draw[thick,mid arrow]  (CFE) --  (PSQM);
    \draw[thick,mid arrow]  (CDS) -- (CDQS);
    \draw[thick,mid arrow]  (PSM) -- (PSQM);
    
    \draw[thick,mid arrow]  (PSM) -- (CDS);
    \draw[thick,mid arrow]  (PSQM) -- (CDQS);
    \draw[thick,mid arrow]  (CDQS) to [out=-160,in=-20] (fR);
    \draw[thick,mid arrow]  (fR) to [out=20,in=160] (CDQS);
    \draw[thick,mid arrow,blue]  (fHbasis) -- (CDQS) node[midway, inner sep=0pt, outer sep=0pt] {\hyperref[lem:fmeasureTOcdqs]{\phantom{\rule{1.25cm}{0.2cm}}}};

    \draw[thick,mid arrow, blue]  (CPauli) -- (CDQS) node[midway, inner sep=0pt, outer sep=0pt] {\hyperref[thm:PaulitoCDQS]{\phantom{\rule{0.2cm}{1.4cm}}}};
    \draw[thick,mid arrow,blue]  (fR) -- (fHbasis)
    node[midway, inner sep=0pt, outer sep=0pt] {\hyperref[lem:frouteTOfmeasure]{\phantom{\rule{1.25cm}{0.2cm}}}};
    
    \draw[thick,mid arrow,blue]  (Cphase) to [out=20,in=160] (CPauli);
    \draw[thick,mid arrow,blue]  (CPauli) to [out=-160,in=-20] (Cphase);

    \draw[draw = none] (Cphase) -- (CPauli) node[midway, inner sep=0pt, outer sep=0pt, yshift= 2.05ex] {\hyperref[thm:cfphaseTOcfpauli]{\phantom{\rule{1.75cm}{0.2cm}}}};

    \draw[draw = none] (CPauli) -- (Cphase) node[midway, inner sep=0pt, outer sep=0pt, yshift= -2.05ex] {\hyperref[thm:cfphaseTOcfpauli]{\phantom{\rule{1.75cm}{0.2cm}}}};
    
    \filldraw[color=white, fill=white] (PSM) ellipse (1.6 and 0.5);
    \filldraw[color=white, fill=white] (CFE) ellipse (1.6 and 0.5);
    \filldraw[color=white, fill=white] (PSQM) ellipse (1.6 and 0.5);
    \filldraw[color=white, fill=white] (CDS) ellipse (1.6 and 0.5);
    \filldraw[color=white, fill=white] (CSWAP) ellipse (1.6 and 0.5);
    \filldraw[color=white, fill=white] (fHbasis) ellipse (1.6 and 0.5);
    \filldraw[color=white, fill=white] (fR) ellipse (1.7 and 0.7);
    \filldraw[color=white, fill=white] (Cphase) ellipse (1.6 and 0.5);
    \filldraw[color=white, fill=white] (CDQS) ellipse (1.6 and 0.5);
    \filldraw[color=white, fill=white] (CPauli) ellipse (1.6 and 0.5);

    \node at (CDQS) {\hyperref[def:CDQS]{CDQS}};
    \node at (CDS) {CDS};
    \node at (PSQM) {PSQM};
    \node at (CFE) {CFE};
    \node at (PSM) {PSM};
    \node at (CPauli) {\hyperref[def:coherentU]{$Cf$-PAULI}};
    \node at (CSWAP) {\hyperref[def:coherentU]{$Cf$-SWAP}};
    \node at (fHbasis) {\hyperref[def:qubitfbb84]{$f$-measure}};
    \node at (fR) {\hyperref[def:frouting]{$f$-route}};
    \node at (Cphase) {\hyperref[def:coherent_phase]{$Cf$-PHASE}};
    
    \end{tikzpicture}
    \caption{The known reductions among function NLQC tasks and primitives from information-theoretic cryptography. Blue implications are new to this work. Clicking on primitives links to definitions; clicking on blue arrows links to proofs of the corresponding implication. Black arrows are proven in \cite{allerstorfer2024relating}.}
    \label{fig:NLQCmap}
\end{figure}

We explore a variety of NLQC tasks, each with different qualitative features whose relationships to their hardness we explore. 
An overview of the tasks we study and reductions among them appears as \Cref{fig:NLQCmap}. 

One grouping of NLQCs we study are the \emph{classically controlled} protocols, which require implementing a small quantum operation controlled on a large classical computation.
This setting is well studied in the quantum position-verification literature \cite{kent2011quantum,buhrman2014position,asadi2024rank,asadi2025linear,bluhm2022single,escola2025quantum,allerstorfer2025making, lau2011insecurity}. 
There, the hope is that as we increase the size of the classical control, the entanglement cost of the NLQC will go up, even while the quantum resources used by an honest player do not.
The examples of classically controlled NLQCs we study are $f$-route, CDQS, and $f$-measure\footnote{In earlier work this task is sometimes referred to as $f$-BB84.}.
In an $f$-measure task, Alice is given $H^{f(x,y)}\ket{b}_Q$ and the string $x\in \{0,1\}^n$ while Bob is given $y\in\{0,1\}^n$.
Their goal is for both Alice and Bob to output $b$.
We show the following relationship among these settings. 

\begin{theorem}\label{thm:measureandroute}
    $f$-route, $f$-measure and CDQS can all be transformed into one another with constant factor overhead in the entanglement cost.
\end{theorem}

This is stated more formally in the main text using our definition of a reduction among NLQCs. 
The proof leverages the existing result that such a transformation exists between $f$-route and CDQS, and then we combine this with proofs that $f$-route $\Rightarrow$ $f$-measure and $f$-measure $\Rightarrow$ CDQS. 
In \Cref{sec:simplifications} we describe the simplifications of the literature or new results that follow from the equivalence of $f$-route and $f$-measure. 
Indeed, they are the best studied examples of NLQCs, and many properties previously proved for each separately now follow by proving it for only one of the two equivalent primitives. 
For example, $f$-measure gains an upper bound of $2^{O(\sqrt{n\log n})}$ via this connection, and efficient protocols for all functions in the complexity class $\mathsf{Mod}_k\mathsf{L}$. 

We also explore several classes of \emph{coherently controlled} operations, which implement $O(1)$ size gates controlled on inputs that may be in superposition. 
In particular we focus on NLQC implementations of the following three unitaries, 
\begin{align}
    Cf\text{-SWAP}_{AA'BB'} &= \sum_{x,y} \text{SWAP}^{f(x,y)}_{A'B'}\otimes \ketbra{x}{x}_A\otimes \ketbra{y}{y}_B \nonumber \\
    Cf\text{-PHASE}_{AA'BB'} &= \sum_{x,y} (-1)^{f(x,y)} \,\ketbra{x}{x}_A\otimes \ketbra{y}{y}_B \nonumber \\
    Cf\text{-}\mathbf{Z}_{AA'BB'} &= \sum_{x,y} \mathbf{Z}^{f(x,y)}_{A'}\otimes \ketbra{x}{x}_A\otimes \ketbra{y}{y}_B.
\end{align}
Note that by acting with local unitaries before and after implementing the NLQC protocol, we can convert a coherently controlled protocol for a unitary $\mathbf{U}_{A'B'}$ into one for $(\mathbf{V}_A\otimes \mathbf{V}_B) \mathbf{U}_{A'B'} (\mathbf{V}_A^\dagger\otimes \mathbf{V}_B^\dagger)$ for any choice of $\mathbf{V}_{A'}, \mathbf{V}_{B'}$. 
This means we can view the above unitaries as representatives of a class of equivalent NLQCs. 
With this in mind, we will refer also to the $Cf$-PAULI, which requires we implement $Cf$-$\mathbf{Z}$ or any unitary of the form $\mathbf{V}_{A'}\mathbf{Z}_{A'}\mathbf{V}_{A'}^\dagger$ (which includes Pauli $\mathbf{X}$ and $\mathbf{Y}$). 

We prove reductions among the coherently controlled NLQC tasks.
To start out, we observe that a $Cf$-SWAP protocol implies a $Cf$-PHASE protocol for the same function using similar resources.
This follows by running the $Cf$-SWAP operation with the $A'B'$ system chosen to be in the $-1$ eigenstate of the SWAP operator.
We were not able to find an implication from $Cf$-PHASE to $Cf$-SWAP, so $Cf$-SWAP may be stronger. 

In a similar spirit, choosing $A'$ to be in the state $\ket{1}_{A'}$ shows $Cf$-$\mathbf{Z}$ (and hence $Cf$-PAULI) implies $Cf$-PHASE. 
Conversely, we show that a $Cf$-PHASE protocol for the function $f\wedge z_1$ for $z_1$ a single bit gives a protocol for $Cf$-PAULI for the same function. 
We can add controls to the implemented unitary by concatenating with the AND operation further: $CZ$ controlled on function $f$ is implied by a $Cf$-PHASE protocol for $f\wedge z_1\wedge z_2$, $CCZ$ controlled on $f$ is implied by a $Cf$-PHASE protocol for $f\wedge z_1\wedge z_2 \wedge z_3$, etc. 
These observations raise the question of whether $f$ and $f\wedge z_1$ have similar costs in $Cf$-PHASE. 
We show that indeed this is the case: a $Cf$-PHASE protocol for $f\wedge z_1$ costs at most $1$ additional EPR pair compared to $f$.
This establishes that NLQCs for $Cf$-$CC...CZ$ for any constant number of controls are all equivalent under constant entanglement overhead reductions. 

We defined $Cf$-SWAP as a natural choice of coherent NLQC task which is strong enough to imply $f$-route. 
$Cf$-PHASE appears to be weaker than $Cf$-SWAP, but we can ask if it is strong enough to recover the incoherent NLQC tasks. 
We find that indeed $Cf$-PHASE actually suffices to obtain the classically controlled tasks $f$-route, $f$-measure and CDQS. 
We do this by proving an implication from the equivalent task $Cf$-PAULI to CDQS.
We leave open understanding if all (non-trivial) coherent tasks are stronger than all incoherent tasks. 

As a final additional result, we show that in the NLQC context the ability to efficiently implement a unitary that interchanges between two states $\ket{\psi_0}, \ket{\psi_1}$ implies the ability to efficiently distinguish between the two ``Fourier'' states $\ket{\phi_{\pm}}=\frac{1}{\sqrt{2}}(\ket{\psi_0}\pm \ket{\psi_1})$. 
This is inspired by a similar result which holds in the computational complexity context \cite{aaronson2020hardness}.  
Note that in the complexity context the ability to apply interchange efficiently implies the ability to distinguish the Fourier states efficiently, and vice versa. 
We were motivated to study if this result also holds in the NLQC context. 
This is because, for instance, it has been conjectured \cite{may2022complexity} that the complexity of a unitary controls its entanglement cost as an NLQC. 
According to this conjecture, interchange and distinguish should have similar entanglement costs. 
We were able to establish this, but only in one direction.
This is our only result relating classes of NLQCs which are not characterized by Boolean functions; we find this result suggestive that there may be other interesting reductions among state transformations or unitaries, though we do not explore this in depth here. 

%%%%%%%%%%%%%%%%%%%%%%%%%%%%%%%%%%%%%%%%%%%%%%%%%%%%%%%%%%%%%%%%%%%%%%
\section{Background and preliminaries}
%%%%%%%%%%%%%%%%%%%%%%%%%%%%%%%%%%%%%%%%%%%%%%%%%%%%%%%%%%%%%%%%%%%%%%

%%%%%%%%%%%%%%%%%%%%%%%%%%%%%%%%%%%%%%%%%%%%%%%%%%%%%%%%%%%%%%%%%%%%%%
\subsection{Tools from quantum information theory}
%%%%%%%%%%%%%%%%%%%%%%%%%%%%%%%%%%%%%%%%%%%%%%%%%%%%%%%%%%%%%%%%%%%%%%

We label the maximally entangled state on $\mathcal{H}_A\otimes \mathcal{H}_B$ by $\Psi^+_{AB}$. 
When we don't specify the Hilbert spaces, $\Psi^+$ denotes a maximally entangled state on two qubits. 

Let $\mathcal{D}(\mathcal{H}_A)$ be the set of density matrices on the Hilbert space $\mathcal{H}_A$. 
Given two density matrices $\rho$, $\sigma\in \mathcal{D}(\mathcal{H}_A)$,  define the fidelity,
\begin{align}
    F(\rho,\sigma) \equiv \left( \tr\left(\sqrt{\sqrt{\rho}\,\sigma\sqrt{\rho}}\right)\right)^2
\end{align}
which is related to the trace distance $\frac{1}{2}\|\rho-\sigma\|_1$ by the Fuchs-van de Graaf inequalities, 
\begin{align}\label{eq:FVDG}
    1- \sqrt{F(\rho,\sigma)} \leq \frac{1}{2}\|\rho-\sigma\|_1 \leq \sqrt{1-F(\rho,\sigma)}.
\end{align}
The fidelity and trace distance allow us to characterize distances between quantum states. 
To characterize distances between quantum channels, define the diamond norm distance \cite{kitaev2002classical,wilde2013quantum}. 
\begin{definition} Let $\mathbfcal{N}_{B\rightarrow C}, \mathbfcal{M}_{B\rightarrow C}: \mathcal{L}(\mathcal{H}_B)\rightarrow \mathcal{L}(\mathcal{H}_C)$ be quantum channels. 
The \textbf{diamond norm distance} is defined by 
\begin{align}
    \|\mathbfcal{N}_{B\rightarrow C}-\mathbfcal{M}_{B\rightarrow C}\|_\diamond = \sup_{d} \max_{\Psi_{A_dB}}\|\mathbfcal{N}_{B\rightarrow C}(\Psi_{A_dB}) - \mathbfcal{M}_{B\rightarrow C}(\Psi_{A_dB})\|_1
\end{align}
where $\Psi_{A_dB}\in \mathcal{D}(\mathcal{H}_{A_d}\otimes \mathcal{H}_B)$ and $\mathcal{H}_{A_d}$ is a $d$ dimensional Hilbert space. 
\end{definition}
The diamond norm distance has an operational interpretation in terms of the maximal probability of distinguishing quantum channels.

We also make use of the von Neumann entropy and some related quantities and inequalities. 
The von Neumann entropy of a density matrix $\rho_A$ is
\begin{align}
    S(A)_\rho = -\tr(\rho_A\log \rho_A)
\end{align}
where we use base $2$ logarithms here and throughout the paper. 
The conditional quantum entropy is
\begin{align}
    S(A|B)_\rho = S(AB)_\rho - S(B)_\rho.
\end{align}
and the mutual information is
\begin{align}
    I(A:B)_\rho = S(A)_\rho + S(B)_\rho - S(AB)_\rho.
\end{align}
We also use the relative entropy, which is given by
\begin{align}
    D(\rho\|\sigma) = \tr \rho\log \rho - \tr \rho\log \sigma 
\end{align}
when the support of $\rho$ is contained within the support of $\sigma$, and defined to be infinity otherwise. 
Pinsker's inequality relates the relative entropy and the trace distance,
\begin{align}\label{eq:pinsker}
    \|\rho-\sigma \|_1 \leq \sqrt{(2\ln 2 )D(\rho\|\sigma)}.
\end{align}
We use the identity
\begin{align}
    I(A:B)_\rho = D(\rho_{AB}\|\rho_A\otimes \rho_B).
\end{align}

We exploit the following inequality relating conditional entropies in a tri-partite state. 
\begin{theorem}[CIT \cite{renes2009conjectured,berta2010uncertainty}]\label{thm:CIT}
    Let $\ket{\psi}_{REF}$ be an arbitrary tripartite state, with $R$ a single qubit. 
    Consider measurements on the $R$ system that produce a measurement result we store in register ${Z}$. Consider measurements in both the computational and Hadamard basis, and denote the post-measurement state when measuring in the computational basis by $\rho_{ZEF}^C$, and when measuring in the Hadamard basis by $\rho_{ZEF}^H$.  
    Then,
    \begin{align*}
        S(Z|E)_{\rho^C} + S(Z|F)_{\rho^H} \geq 1 \enspace.
    \end{align*}
\end{theorem}

%%%%%%%%%%%%%%%%%%%%%%%%%%%%%%%%%%%%%%%%%%%%%%%%%%%%%%%%%%%%%%%%%%%%%%
\subsection{NLQC reductions}
%%%%%%%%%%%%%%%%%%%%%%%%%%%%%%%%%%%%%%%%%%%%%%%%%%%%%%%%%%%%%%%%%%%%%%

We first define what a NLQC aims to achieve, which is to complete what we call a $2\rightarrow 2$ quantum task. 
\begin{definition}
    A \textbf{$2\rightarrow 2$ quantum task} is defined by a pair of input systems $A$, $B$, a pair of output systems $A'$, $B'$, and a set of input/output state pairs $\mathcal{S}=\{ (\rho_{RAB},\sigma_{RA'B'})\}$. We require that there exists a quantum channel $\mathbfcal{M}_{AB \to A'B'}$ such that $(\mathcal{I}_R \otimes \mathbfcal{M}_{AB\rightarrow A'B'})(\rho_{RAB}) = \sigma_{RA'B'}$ for all $\rho_{RAB}$.
\end{definition}
This definition of task could be further generalized to the case when multiple outputs are allowed per a given input state (i.e., a quantum analogue of a search problem), but we leave that for future exploration.

We will also be interested in \emph{families} of $2\rightarrow 2$ quantum tasks, which are collections of $2\rightarrow 2$ tasks parameterized by a natural number $n$.
The parameter $n$ will correspond to an input size with the exact relation specified in the definition of each family of tasks.  
We will label families of tasks with capital letters $F$, $G$, etc., where these denote sets of tasks, so that $F=\{F_n\}_n$, $G=\{G_n\}_n$, etc. 
As an example, the $f$-route example mentioned in the introduction with a choice of Boolean function family $\{f_n\}_n$ defines a family of $2\rightarrow 2$ task, with each member of the family labelled by an element of $\{f_n\}_n$.  

We are interested in implementing $2\rightarrow 2$ quantum tasks in the form of NLQCs, which we define more carefully next. 
\begin{definition}
    A \textbf{non-local quantum computation (NLQC)} is a channel in the form
    \begin{align}
        \mathbfcal{N}_{AB\rightarrow A'B'}(\,\cdot\,) =(\mathbfcal{W}_{K_aM_a\rightarrow A'}\otimes\mathbfcal{W}_{K_bM_b\rightarrow B'})\circ (\mathbfcal{V}_{AL\rightarrow K_{a}M_b}\otimes \mathbfcal{V}_{RB\rightarrow M_aK_b})(\,\cdot\, \otimes\Psi_{LR}).
    \end{align}
    We refer to $\Psi_{LR}$ as the resource system, the $\mathbfcal{V}$ channels as the first round operations, and the $\mathbfcal{W}$ channels as the second round operations.
    
    We say a NLQC is an $\epsilon$-correct implementation of a $2\rightarrow 2$ task $F_n$ if the channel $\mathbfcal{N}_{AB\rightarrow A'B'}$ implemented as a NLQC is $\epsilon$-close in diamond norm to the channel $\mathbfcal{M}_{AB\rightarrow A'B'}$ relating the input and output states in the definition of the $2\rightarrow 2$ task.
\end{definition}
We also introduce the following convention.  
When considering a specific $2\rightarrow 2$ task, call it $X$, and referring to an implementation of this as a non-local quantum computation, we will write ``the non-local computation $X$''. 
This can be understood as shorthand for ``implementations of the $2\rightarrow 2$ task $X$ in the form of a non-local quantum computation''. 

We are interested in understanding the relative difficulty of implementing different $2\rightarrow 2$ tasks as NLQCs.
To do this, we will define a notion of reduction between $2\rightarrow 2$ tasks. 
Heuristically, our notion of reduction says that the task $G$ reduces to $F$ when (a few copies of) the resources to implement $F$ as a NLQC can be used to implement $G$ as a NLQC. 

To formalize our notion of a reduction among $2\rightarrow 2$ tasks, it is helpful to recall the notion of reduction among computational problems. 
There we say a function family $A=\{a_n\}_n$ is polynomial-time Turing reducible to function family $B=\{b_n\}$ if a poly$(n)$ time machine with oracle access to $B$ can solve $A$. 
More generally, we can replace `poly-time' with any other complexity class, call it $\mathcal{X}$; the notion of reduction is most meaningful when $\mathcal{X}$ is itself too weak to implement $A$ or $B$.
Inspired by this definition, we give the following definition of reduction among NLQC classes. 
\begin{definition}\label{def:reduction}
    Let $F$ and $G$ be families of $2\rightarrow 2$ task, and $\alpha$, $\beta$, $\delta$ all be functions of $(n,\epsilon)$, with $\delta\rightarrow0$ as $\epsilon\rightarrow 0$.
    Then we say there is a $(\alpha, \beta,\delta)$-\textbf{reduction} from $G$ to $F$ if, for any resource state $\ket{\Psi^{n,\epsilon}}_{LR}$ which implements $F_n$ at least $\epsilon$-correctly, it is possible to implement $G_n$ $\delta$-correctly using $\alpha$ copies of $\ket{\Psi^{n,\epsilon}}_{LR}$ along with $\beta$ additional qubits of shared resource state.
\end{definition}
When it is necessary to distinguish this notion of a reduction from other similar notions (e.g. \Cref{def:oraclereduction} below), we refer to it as a \emph{resource state reduction}. 
When there is a reduction from family $F$ to family $G$, we will also say there is an \textbf{implication} from $G$ to $F$, and write $G \Rightarrow F$.

Regarding the values of $\alpha(n)$ and $\beta(n)$, in practice we will find reductions that have $\alpha(n)=O(1)$, $\beta(n)=O(1)$. 
We define an $O(1)$ reduction to be one where $\alpha=O(1)$, $\beta=O(1)$. 
Weaker reductions that rely on larger values of $\alpha(n), \beta(n)$ can also be meaningful, at least when $\beta(n)$ is smaller than the number of qubits of resource needed to implement $F$ and $G$.
The requirement on $\delta$ that $\delta \rightarrow 0$ as $\epsilon\rightarrow 0$ is needed to ensure our notion of reduction is not trivial. 
This is because many $2\rightarrow 2$ tasks can be implemented with some non-zero accuracy trivially, without using any resource state, but become non-trivial at a certain threshold of accuracy. 
For instance in $f$-route, the quantum system that is being routed can be brought left or right at random, completing the task correctly half the time.
Thus a construction showing that some other $2\rightarrow 2$ task can be used to complete $f$-route with probability $1/2$ is not meaningful, since we don't need to use any resource state at all. 
Requiring $\delta \rightarrow 0$ when $\epsilon\rightarrow 0$ ensures the resource state for $F$ is strong enough, at least if $F$ is implemented with high correctness, to implement a highly correct (and hence non-trivial) $G$. 

Other notions of reduction among $2\rightarrow 2$ tasks are also possible, which highlight different aspects of their structure.
For instance, we have not required that our resource state reductions are communication efficient, nor that they involve small computational overheads. 
Both aspects of NLQC could be explored with other notions of reduction.

One further notion of reduction that we make use of requires that the protocol for $G_n$ be given by using the protocol for $F_n$ used as an oracle, meaning it has access only to copies of the implementation of $F_n$ but not access directly to the state $\ket{\Psi^{n,\epsilon}}_{LR}$. 
We give a definition of this notion of reduction next. 
\begin{definition}\label{def:oraclereduction}
    Let $F$ and $G$ be families of $2\rightarrow 2$ task, and $\alpha$, $\beta$, $\delta$ all be functions of $(n,\epsilon)$, with $\delta\rightarrow0$ as $\epsilon\rightarrow 0$.
    Then we say that there is a $(\alpha, \beta,\delta)$ \textbf{oracle reduction} from $G$ to $F$ if $G$ can be implemented $\delta$ correctly by using $\alpha$ parallel implementations of $F$ along with $\beta$ additional qubits of resource system and communication. 
\end{definition}
The structure of an implementation of $G$ using $F$ as an oracle is shown in \Cref{fig:oraclereduction}. 

\begin{figure}
    \centering
    \begin{tikzpicture}[scale=0.5]

    \draw[thick, fill=gray,opacity=0.5] (-3,-1.5) -- (3,-1.5) -- (3,1.5) -- (-3,1.5) -- (-3,-1.5);
    
    %lower left box
    \draw[thick] (-5,-5) -- (-5,-3) -- (-3,-3) -- (-3,-5) -- (-5,-5);
    
    %lower right box
    \draw[thick] (5,-5) -- (5,-3) -- (3,-3) -- (3,-5) -- (5,-5);
    
    %top right box
    \draw[thick] (5,5) -- (5,3) -- (3,3) -- (3,5) -- (5,5);
    
    %top left box
    \draw[thick] (-5,5) -- (-5,3) -- (-3,3) -- (-3,5) -- (-5,5);
    
    %left vertical wire
    \draw[thick,mid arrow,dashed] (-4.5,-3) -- (-4.5,3);
    
    %right vertical wire
    \draw[thick,mid arrow,dashed] (4.5,-3) -- (4.5,3);
    
    %left to right wire
    \draw[thick,mid arrow, dashed] (-4,-3) -- (-4,-1.5) to [out=90,in=-90] (4,3);
    
    %right to left wire
    \draw[thick,mid arrow, dashed] (4,-3) -- (4,-1.5) to [out=90,in=-90] (-4,3);

    \draw[thick,gray] (-3.5,-3) -- (-2,-1.5);
    \draw[thick,gray] (3.5,-3) -- (2,-1.5);
    \draw[thick,gray] (3.5,3) -- (2,1.5);
    \draw[thick,gray] (-3.5,3) -- (-2,1.5);
    
    %entanglement
    \draw[thick] (-3.5,-6) -- (3.5,-6) -- (0,-8) -- (-3.5,-6);
    \draw[thick] (-3.25,-6) -- (-3.25,-5);
    \draw[thick] (3.25,-6) -- (3.25,-5);
    \node at (0,-7) {$\Phi$};

    \node at (0,-0.25) {$F_n^{\otimes \alpha(n)}$};
    
    %input wires
    \draw[thick] (-4.5,-6) -- (-4.5,-5);
    \draw[thick] (4.5,-6) -- (4.5,-5);
    
    %output wires
    \draw[thick] (4.5,5) -- (4.5,6);
    \draw[thick] (-4.5,5) -- (-4.5,6);
    
    \end{tikzpicture}
    \caption{An oracle reduction from a NLQC $G$ to $F$. The protocol for $G_n$ uses $\alpha(n)$ instances of $F_n$, plus a $\beta(n)$ qubit resource system $\Phi$ and $\beta(n)$ qubits of communication. }
    \label{fig:oraclereduction}
\end{figure}

Compared to the notion of resource state reduction in \Cref{def:reduction}, an oracle reduction gives a tighter relationship between task families. 
This is because in a resource state reduction, we've shown that a resource state that works for $F$ also (up to certain overheads) works for $G$, while an oracle reduction implies not only that, but that additionally the local operations used in $F$ can also be applied to complete $G$. 
Resource state reductions allow greater flexibility in what operations are performed locally to exploit the non-local resources, and reductions under this definition focus on the power of non-local resources for completing NLQCs. 
For example, some of the reductions described by \Cref{thm:measureandroute} are not oracle reductions. 
Note also that if there is an oracle reduction from $F$ to $G$ then it immediately follows that there is a resource state reduction from $F$ to $G$ in the sense of \Cref{def:reduction}. 
When not further specified `reduction' in this article will always mean resource state reduction. 

A useful fact is that oracle reductions have convenient error-propagation properties. 
This follows from the next remark.
\begin{remark}\label{remark:blackboxerror}
    Suppose that there is a protocol which uses parallel implementations of channels $\mathbfcal{N}_1, \mathbfcal{N}_2,...$ to execute some target channel $\mathbfcal{N}$, so that
    \begin{align}
        \mathbfcal{N} = \mathbfcal{F}\circ \left(\bigotimes_{i=1}^m \mathbfcal{N}_i \right) \circ \mathbfcal{P}.
    \end{align}
    Then suppose we replace the exact implementations of the $\mathbfcal{N}_i$ with approximate implementations $\overline{\mathbfcal{N}}_i$ which satisfy $\Vert \mathbfcal{N}_i-\overline{\mathbfcal{N}}_i\Vert_\diamond \leq \epsilon_i$. 
    Then
    \begin{align}
        \overline{\mathbfcal{N}} = \mathbfcal{F}\circ \left(\bigotimes_{i=1}^m \overline{\mathbfcal{N}}_i \right) \circ \mathbfcal{P}
    \end{align}
    is $\sum_i \epsilon_i$ close in diamond norm to $\mathbfcal{N}$. 
\end{remark}
This remark follows from the properties of the diamond norm. 

To apply this remark to our oracle reductions, note that any realization of $G_n$ using oracle instances of $F_n$ must be exactly correct when $F_n$ is exactly correct. 
This is because we require that $\delta \rightarrow 0$ as $\epsilon\rightarrow 0$. 
When we instead use $\epsilon$ correct realizations of $F_n$, we get from the remark above that the diamond norm distance to the perfectly correct channel is at most $\alpha \cdot \epsilon$, where $\alpha $ is the number of instances of $F_n$ used. 

Finally, we will also discuss sets of $2\rightarrow 2$ tasks which are larger than the families introduced above. 
These sets have an additional parameter aside from the input size. 
For example, we can discuss the family of tasks given by routing on a particular function family $\{f_n\}_n$, but we can also go further and consider the set of tasks consisting of routing for any possible choice of Boolean function $f$.  
We refer to these larger sets of tasks as classes of $2\rightarrow 2$ task or classes of NLQCs. 
In this paper these classes are parameterized either by a choice of Boolean function, as in $f$-routing, or a choice of pair of quantum states. 
Other sets parameterized differently could also be defined. 

%%%%%%%%%%%%%%%%%%%%%%%%%%%%%%%%%%%%%%%%%%%%%%%%%%%%%%%%%%%%%%%%%%%%%%
\section{Some classes and their properties}
%%%%%%%%%%%%%%%%%%%%%%%%%%%%%%%%%%%%%%%%%%%%%%%%%%%%%%%%%%%%%%%%%%%%%%

In this section we collect definitions of NLQC tasks that we study here. 
Most of these tasks have been studied elsewhere in the literature. 

%%%%%%%%%%%%%%%%%%%%%%%%%%%%%%%%%%%%%%%%%%%%%%%%%%%%%%%%%%%%%%%%%%%%%%
\subsection{Classically controlled operations}
%%%%%%%%%%%%%%%%%%%%%%%%%%%%%%%%%%%%%%%%%%%%%%%%%%%%%%%%%%%%%%%%%%%%%%
 
First we define $f$-route, introduced in \cite{kent2011quantum, buhrman2013garden}. 
\begin{definition}\label{def:frouting}
    An \textbf{$f$-route} NLQC task is defined by a choice of Boolean function $f:\{ 0,1\}^{2n}\rightarrow \{0,1\}$, and a $d_Q$ dimensional Hilbert space $\mathcal{H}_Q$.
    Inputs $x\in \{0,1\}^{n}$ and system $Q$ are given to Alice, and input $y\in \{0,1\}^{n}$ is given to Bob.
    Alice and Bob exchange one round of communication, with the combined systems received or kept by Bob labelled $M$ and the systems received or kept by Alice labelled $M'$.
    Label the combined actions of Alice and Bob in the first round as $\mathbfcal{N}^{x,y}_{Q\rightarrow MM'}$. 
    The $f$-route task is completed $\epsilon$-correctly if there exists a family of channels $\mathbfcal{D}^{x,y}_{M\rightarrow Q}$ such that,
    \begin{align}
        \forall (x,y)\in X\times Y \,\,\, s.t. \,\, f(x,y)=1,\,\,\, \|\mathbfcal{D}^{x,y}_{M\rightarrow Q} \circ\tr_{M'} \circ\mathbfcal{N}^{x,y}_{Q\rightarrow MM'} -\mathbfcal{I}_{Q\rightarrow Q}\|_\diamond \leq \epsilon
    \end{align}
    and there exists a family of channels $\mathbfcal{D}^{x,y}_{M'\rightarrow Q}$ such that
    \begin{align}
        \forall (x,y)\in X\times Y \,\,\, s.t. \,\, f(x,y)=0,\,\,\,\|\mathbfcal{D}^{x,y}_{M'\rightarrow Q} \circ\tr_{M} \circ\mathbfcal{N}^{x,y}_{Q\rightarrow MM'} -\mathbfcal{I}_{Q\rightarrow Q}\|_\diamond \leq \epsilon.
    \end{align}
    In words, Bob can recover $Q$ if $f(x,y)=1$ and Alice can recover $Q$ if $f(x,y)=0$. 
\end{definition}

\begin{remark}\label{remark:routeandSWAP}
$f$-route is equivalent to $f$-SWAP, where in $f$-SWAP the goal is to implement SWAP$_{A_0B_0}^{f(x,y)}$ with $A_0$ on Alice's side and $B_0$ on Bob's side. 
We take the error to be the diamond norm distance between the implemented protocol and implementing SWAP$_{A_0B_0}^{f(x,y)}$ exactly.
We can implement $f$-SWAP by first teleporting system $B_0$ to Alice, and then playing $f$-route on $A_0$ and $\neg f$-route on $B_0$. This costs twice the resources used for $f$- route, and $O(1)$ extra entanglement for the teleportation of $B_0$.
\\
The maximal error in the two $f$-route protocols used gives an upper bound on the $f$-SWAP error.
$f$-SWAP immediately gives $f$-route with the same error by taking $B_0$ to be in any fixed state.
\end{remark}

$f$-route was shown to be related to CDQS in \cite{allerstorfer2024relating}; in our language their proof shows a $O(1)$ reduction from $f$-route to CDQS and vice versa. 
We define CDQS next, first defined in the quantum context in \cite{asadi2024conditional}. 
\begin{definition}\label{def:CDQS}
    A \textbf{conditional disclosure of secrets} task with quantum resources (CDQS) is defined by a choice of function $f:\{0,1\}^{2n}\rightarrow \{0,1\}$, and a $d_Q$ dimensional Hilbert space $\mathcal{H}_Q$ which holds the secret.
     The task involves inputs $x\in \{0,1\}^{n}$ and system $Q$ given to Alice, and input $y\in \{0,1\}^{n}$ given to Bob.
    Alice sends message system $M_a$ to the referee, and Bob sends message system $M_b$. 
    Label the combined message systems as $M=M_aM_b$.
    Label the quantum channel defined by Alice and Bob's combined actions $\mathbfcal{N}_{Q\rightarrow M}^{x,y}$. 
    We put the following two conditions on a CDQS protocol. 
    \begin{itemize}
        \item $\epsilon$\textbf{-correct:} There exists a channel $\mathbfcal{D}^{x,y}_{M\rightarrow Q}$, called the decoder, such that
        \begin{align}
            \forall (x,y)\in X\times Y \,\,\, s.t. \,\, f(x,y)=1,\,\,\, \|\mathbfcal{D}^{x,y}_{M\rightarrow Q}\circ \mathbfcal{N}^{x,y}_{Q\rightarrow M} - \mathbfcal{I}_{Q\rightarrow Q}\|_\diamond \leq \epsilon.
        \end{align}
        \item $\delta$\textbf{-secure:} There exists a quantum channel $\mathbfcal{S}_{\varnothing \rightarrow M}^{x,y}$, called the simulator, such that
        \begin{align}
            \forall (x,y)\in X\times Y \,\,\, s.t. \,\, f(x,y)=0,\,\,\, \|\mathbfcal{S}_{\varnothing \rightarrow M}^{x,y} \circ \tr_Q - \mathbfcal{N}_{Q\rightarrow M}^{x,y}\|_\diamond \leq \delta.
        \end{align}
    \end{itemize}
    We call the log-dimension of $M$ the communication cost of the CDQS protocol. 
    We call the log-dimension of the (possibly entangled) state shared by Alice and Bob at the beginning of the protocol the entanglement cost of the protocol. 
\end{definition}
The correctness parameter $\epsilon$ captures how well the secret can be recovered by the referee in $f^{-1}(1)$ instances; the security parameter $\delta$ captures how close the message $M$ (seen by the referee) is to a state that doesn't depend on the input at all. Thus when $\delta=0$, the message system is a fixed state independent of the secret, and hence doesn't reveal the secret. 

CDQS was introduced in \cite{allerstorfer2024relating}, where its equivalence to $f$-route was shown.
Note that CDQS is not formally an NLQC task, though it is closely related and serves as a useful intermediate setting in our proofs. 
The same reference also showed that classical CDS protocols also give CDQS protocols, which allowed results on classical CDS to be applied to $f$-route.  
CDQS was studied further in \cite{asadi2024conditional, girish2025comparing}, where a number of properties are established. 
One property which we will make use of is amplification, stated formally next, which allows the $\epsilon, \delta$ parameters in a given CDQS protocol to be reduced.  
\begin{theorem}\label{thm:CDQSamplification}\textbf{CDQS amplification:}
    Let $F_Q$ be a $\CDQS$ for a function $f$ that supports one qubit  secrets with correctness error $\delta=0.09$ and privacy error $\epsilon=0.09$, has communication cost $c$, and entanglement cost $E$. 
    Then for every integer $k$, there exists a $\CDQS$ $G_Q$ for $f$ with $k$-qubit secrets, privacy and correctness errors of $2^{-\Omega(k)}$, and communication and entanglement complexity of size $O(k c)$ and $O(k E)$, respectively. 
\end{theorem}
This theorem was proven in \cite{asadi2024conditional}.

We will make use of a variant of CDQS where the secret is assumed to be classical, which we label CDQS'.
Since our proofs will relate other primitives to CDQS' (and then CDQS' to CDQS), we give a careful definition of CDQS' as well. 
\begin{definition}\label{def:CDQSprime}
    A \textbf{conditional disclosure of secrets} task with quantum resources and classical secret (CDQS') is defined by a choice of function $f:\{0,1\}^{2n}\rightarrow \{0,1\}$ and a classical variable $S\in \{0,1\}^{|s|}$ which records the secret.
    The task involves inputs $x\in \{0,1\}^{n}$ and system $Q$ given to Alice, and input $y\in \{0,1\}^{n}$ given to Bob.
    Alice sends message system $M_a$ to the referee, and Bob sends message system $M_b$. 
    Label the combined message systems as $M=M_aM_b$.
    Label the quantum channel defined by Alice and Bob's combined actions $\mathbfcal{N}_{Q\rightarrow M}^{x,y}$. 
    We put the following two conditions on a CDQS' protocol. 
    \begin{itemize}
        \item $\epsilon$\textbf{-correct:} There exists a channel $\mathbfcal{D}^{x,y}_{M\rightarrow S}$, called the decoder, which outputs a classical string in $S$ such that
        \begin{align}
            \forall (x,y)\in X\times Y \,\,\, s.t. \,\, f(x,y)=1,\,\,\,\forall s\in S,\,\,\, \textnormal{Pr} \left[\mathbfcal{D}^{x,y}_{M\rightarrow S}(\rho_M(x,y,s)) =s\right] \geq 1-\epsilon.
        \end{align}
        \item $\delta$\textbf{-secure:} There exists a family of density matrices $\{\sigma_M(x,y)\}_{x,y}$, called the simulator distribution, such that
        \begin{align}
            \forall (x,y)\in X\times Y \,\,\, s.t. \,\, f(x,y)=0,\,\,\,\forall s\in S,\,\,\, \left\Vert \sigma_M(x,y) - \rho_M(s,x,y) \right\Vert_1 \leq \delta.
        \end{align}
    \end{itemize}
\end{definition}
Entanglement and communication costs of CDQS and CDQS' are equivalent up constant factor overheads; the proof is similar to the proof of Theorem 22 in \cite{allerstorfer2024relating}. 
We adapt their proof to our setting here and give it for completeness in \Cref{sec:CDQSproperties}. 
The two directions of the equivalence are given as \Cref{thm:CDQStoCDQS'} and \Cref{thm:CDQS'toCDQS}.

We will make use of the following simple property of CDQS' protocols. 
\begin{theorem}
    Suppose we have a CDQS' protocol which hides a uniformly random classical secret $r$. Then, using the same resource state but adding $|r|$ bits of communication, we can build a protocol that hides a secret $s$ with $|s|=|r|$ and an arbitrary distribution $P_S$. 
\end{theorem}
This is proved in \Cref{sec:CDQSproperties}. 

Our contribution to the study of these NLQC classes will be to show that CDQS and $f$-route are equivalent under $O(1)$ reductions to $f$-measure, defined next. 

\begin{definition}\label{def:qubitfbb84}
    A \textbf{$f$-measure} task is defined by a choice of Boolean function $f:\{ 0,1\}^{2n}\rightarrow \{0,1\}$, and a $2$ dimensional Hilbert space $\mathcal{H}_Q$.
    Inputs $x\in \{0,1\}^{n}$ and system $Q$ are given to Alice, and input $y\in \{0,1\}^{n}$ is given to Bob.
    The system $Q$ is in the maximally entangled state with a reference system $R$.\footnote{Notice that compared to the description of $f$-measure given in the introduction, we now have the referee give Alice and Bob one end of a maximally entangled state and then later measure the reference system. 
    The referee's measurement is chosen such that the post-measurement state is one of the BB84 states, coinciding with the description in the introduction. 
The two formulations are equivalent.} 
    Alice and Bob exchange one round of communication, with the combined systems received or kept by Alice labelled $M$ and the systems received or kept by Bob labelled $M'$.
    Define projectors
    \begin{align*}
        \Pi^{q,b}=H^q\ketbra{b}{b}H^q \enspace.
    \end{align*}
    The referee will measure $\{\Pi^{f(x,y),0},\Pi^{f(x,y),1}\}$ on the $R$ system and find measurement outcome $b\in \{0,1\}$. 
    The qubit $f$-measure task is completed $\epsilon$-correctly on input $(x,y)$ if Alice and Bob both output $b$. 
    More formally, Alice and Bob succeed if there exist POVM's $\{\Lambda^{x,y,0}_{M},\Lambda^{x,y,1}_{M}\}$, $\{\Lambda^{x,y,0}_{M'},\Lambda^{x,y,1}_{M'}\}$ such that,
    \begin{align}
        \sum_b \tr(\Pi^{f(x,y),b}_R\otimes \Lambda^{x,y,b}_{M} \otimes \Lambda^{x,y,b}_{M'} \rho_{RMM'}) \geq 1-\epsilon \enspace.
    \end{align}
\end{definition}
$f$-measure is well studied as a candidate protocol for secure quantum position-verification \cite{kent2011quantum, buhrman2014position, asadi2024rank, asadi2025linear, bluhm2022single, escola2025quantum,allerstorfer2025making}.

A number of properties of $f$-measure, $f$-route, and CDQS have been established in earlier literature. 
Regarding lower bounds, \cite{asadi2024rank} showed that when $\epsilon=0$, both $f$-route and $f$-measure require resource states with Schmidt rank at least the rank of $f(x,y)$ viewed as a $2^n\times 2^n$ matrix. 
For both $f$-route and $f$-measure, \cite{bluhm2022single} showed that the memory register to complete these tasks for random choices of function $f$ must be linear in $n$, though the proof left open if this needed to be a quantum memory or if classical memory sufficed.
Again for both $f$-route and $f$-measure, \cite{asadi2025linear} showed lower bounds on the number of quantum gates needed which are linear in the communication complexity of $f$.
For both $f$-route and $f$-measure, \cite{escola2025quantum} showed a parallel repetition property. 
Some of the lower bounds given in \cite{asadi2024conditional} were not known for $f$-measure (though we know of no attempt to extend them to that case). 
Thus regarding lower bounds, nearly all properties of $f$-route and $f$-measure matched, but were always proven separately and using distinct strategies.

Regarding upper bounds the state of knowledge for $f$-route and $f$-measure differed. 
For instance \cite{cree2023code} showed an upper bound on $f$-route in terms of the minimal size of a span program computing $f(x,y)$; this allows the complexity class $\mathsf{Mod}_k\mathsf{L}$ to be performed with polynomial entanglement.\footnote{It was later realized that combining the implications classical CDS $\Rightarrow$ CDQS $\Rightarrow$ $f$-route from \cite{allerstorfer2024relating} with known upper bounds on classical CDS also yields the same upper bound as was proven in \cite{cree2023code}.} 
Meanwhile, for $f$-measure, the best upper bound was $2^{M(f)}$ where $M(f)$ is the memory to compute $f(x,y)$ on a Turing machine \cite{buhrman2013garden}; this allowed $f$-measure instances with $f \in \mathsf{L}\subseteq \mathsf{Mod}_k\mathsf{L}$ to be completed efficiently.
For generic functions, \cite{allerstorfer2024relating} employed a result from the CDS literature to show an upper bound of $2^{O(\sqrt{n\log n})}$ for $f$-route but no such bound was known for $f$-measure.
For $f$-route, one function (quadratic residuosity) believed to be outside $\mathsf{P}$ was known to have an efficient protocol, but no such example was known for $f$-measure.

%%%%%%%%%%%%%%%%%%%%%%%%%%%%%%%%%%%%%%%%%%%%%%%%%%%%%%%%%%%%%%%%%%%%%%
\subsection{Coherently controlled operations}
%%%%%%%%%%%%%%%%%%%%%%%%%%%%%%%%%%%%%%%%%%%%%%%%%%%%%%%%%%%%%%%%%%%%%%

Next we define our coherently controlled NLQCs, beginning with implementing a coherent phase.
This NLQC has been studied previously in \cite{junge2022geometry}. 

\begin{definition} \label{def:coherent_phase} A \textbf{coherent phase} task is defined by a Boolean function, $f: \{0,1\}^{2n}  \rightarrow \{0,1\}$. 
System $A$ consisting of $n$ qubits is given to Alice, and system $B$ consisting of $n$ qubits is given to Bob.
Alice and Bob's goal is to implement the unitary
\begin{align}
    Cf\textnormal{-PHASE}_{AB} =  \sum_{x,y \in \{0,1\}^n}(-1)^{f(x,y)} \ketbra{x}{x}_A \otimes \ketbra{y}{y}_B.
\end{align}
in the form of a NLQC, with $A$ starting and ending on Alice's side and $B$ starting and ending on Bob's side. 
We say they have implemented the coherent phase operation with correctness $\epsilon$ if their implementation is $\epsilon$-close to $Cf$-\textnormal{PHASE} in diamond norm distance.
\end{definition}
We will be most interested in the case where the phase is $-1$, but one could define a similar task with other choices of phase.
One motivation for considering this NLQC in particular is that these are the NLQCs discussed in \cite{junge2022geometry}, for which strong lower bounds can be proven under the assumption of some mathematical conjectures.
Any NLQC classes which imply $Cf$-PHASE also inherit these conditional bounds. 

In relating the coherent phase to other NLQC task, it will be helpful to establish some key properties of coherent phase. 
We give two observations, both proven in \Cref{sec:phaseproperties}.

\begin{restatable}{theorem}{reoplusphase}\label{thm:oplusphase}
    Given a coherent phase protocol for $f_1$ which uses resource system $\Psi_1$, and a coherent phase protocol for $f_2$ using resource system $\Psi_2$, there is a protocol for $f_1\oplus f_2$ using resource system $\Psi_{1}\otimes \Psi_2$. Furthermore, the construction is an oracle reduction so that an $\epsilon_1$ correct protocol for $f_1$ and $\epsilon_2$ correct protocol for $f_2$ leads to a $\epsilon_1+\epsilon_2$ correct protocol for $f_1\oplus f_2$. 
\end{restatable}

\begin{restatable}{theorem}{rePhaseANDclosure}\label{thm:PhaseANDclosure}
    Suppose we have a coherent phase protocol for $f(x,y)$ using resource system $\Psi$. Then, there is a coherent phase protocol for $f(x,y)\wedge z$ using $\Psi\otimes \Psi^+$ where $\Psi^+$ is the maximally entangled state of two qubits. 
    Furthermore, the protocol is an oracle construction and uses a single implementation of $f(x,y)$ so that if the error in the $f(x,y)$ implementation is $\epsilon$, then the error in implementing $f(x,y)\wedge z$ is also $\epsilon$.
\end{restatable} 

In fact, we extend this to a more general property.
To state it, we first need the following definition. 
\begin{definition}\label{def:rank}
    Consider a function $f:X\times Y\rightarrow \{0,1\}$ with $X=\{0,1\}^n$, $Y=\{0,1\}^n$.
    The \textbf{rank} of $f$ is the smallest integer $m$ such that there exist functions $\{h_i(x)\}$, $\{h_i'(y)\}$ satisfying
    \begin{align}
        f(x,y)= \bigoplus_{i=1}^m h_i(x)\wedge h'_i(y).
    \end{align}
\end{definition}
With this definition we can state the following theorem. 

\begin{restatable}{theorem}{rephaseclosure}\label{thm:phaseclosure}
Consider functions $f$ and $g$, and let the rank of $g$ be denoted $m$.
Then there is a $O(m)$ reduction from a coherent phase protocol for $f\wedge g$ to a coherent phase protocol for $f$. 
Furthermore, the construction is an oracle reduction and uses $m$ invocations of the protocol for $f$, so that the error in implementing $f\wedge g$ is $m\cdot \epsilon$ where $\epsilon$ is the error in the protocol for $f$. 
\end{restatable}
We give the proof in \Cref{sec:phaseproperties}. 

One related class of NLQCs is the following.
\begin{definition} \label{def:coherentU} For a unitary $\mathbf{U}_{A'B'}$, a \textbf{coherent $\mathbf{U}$} task is defined by a Boolean function, $f: \{0,1\}^{2n}\rightarrow \{0,1\}$. 
The goal is to implement the unitary
    \begin{align}
        Cf\text{-}\mathbf{U} =  \sum_{x,y \in \{0,1\}^n} \mathbf{U}_{A'B'}^{f(x,y)}\otimes \ketbra{x}{x}_A \otimes \ketbra{y}{y}_B
    \end{align}
in the form of a NLQC, with $AA'$ beginning and ending with Alice and $BB'$ beginning and ending with Bob.
We say the implementation is $\epsilon$-correct if it is $\epsilon$-close to the above unitary in diamond norm distance.
\end{definition}
$Cf$-SWAP and $Cf$-PAULI are special cases of $Cf$-$\mathbf{U}$, where we choose $\mathbf{U}_{A'B'}=\textnormal{SWAP}_{A'B'}$ and $\mathbf{U}_{A'B'}=\mathbf{Z}_{A'}\otimes \mathcal{I}_{B'}$ respectively.

\begin{remark}\label{remark:coherenttoincoherent}
    A protocol for coherently controlled unitary $\mathbf{U}$ immediately gives a protocol for a classical controlled unitary $\mathbf{U}$, by recording the classical input strings $x,y\in\{0,1\}^n$ as quantum states $\ket{x}_A, \ket{y}_B$ and taking them as input to the coherent protocol. For example, $Cf$-SWAP implies $f$-SWAP (and hence $f$-route by \Cref{remark:routeandSWAP}).
\end{remark}

Similar to coherent phase, the coherent unitary NLQC also has some nice composition properties. 
We give the following two properties, both proven in \Cref{sec:coherentUproperties}. 
\begin{restatable}{theorem}{retensorproduct}\label{thm:tensorproduct}
\textbf{Tensor product parallel composition:} Suppose we can implement coherently controlled $\mathbf{U}_1$ and $\mathbf{U}_2$, controlled on function $f$, using resource systems $\Psi^1$ and $\Psi^2$ respectively. 
Then we can apply $\mathbf{U}^f_1\otimes \mathbf{U}^f_2$ coherently using $\Psi^1\otimes \Psi^2$. Furthermore, the construction isan oracle reduction so that the overall error is the sum of the error for implementing $\mathbf{U}_1$ and $\mathbf{U}_2$.
\end{restatable}

\begin{restatable}{theorem}{repowers}\label{thm:powers} 
    \textbf{Commuting unitaries:} Suppose we can implement unitaries $\mathbf{U}^1_{A'}$ and $\mathbf{U}^2_{A'}$ with $A'$ on Alice's side, coherently controlled on function $f$ using resource systems $\Psi^1$ and $\Psi^2$ respectively, and that $\mathbf{U}^1_{A'}$ and $\mathbf{U}_{A'}^2$ commute.  
    Then we can apply $\mathbf{U}^1_{A'}\mathbf{U}^2_{A'}$ coherently controlled off of the same function using resource system $\Psi^1\otimes \Psi^2$.
    Furthermore, the construction is an oracle reduction so that the overall error is the sum of the errors in the implementations of each of $\mathbf{U}^1_{A'}$ and $\mathbf{U}^2_{A'}$.
\end{restatable}

%%%%%%%%%%%%%%%%%%%%%%%%%%%%%%%%%%%%%%%%%%%%%%%%%%%%%%%%%%%%%%%%%%%%%%
\subsection{State and unitary classes}\label{sec:stateandunitary}
%%%%%%%%%%%%%%%%%%%%%%%%%%%%%%%%%%%%%%%%%%%%%%%%%%%%%%%%%%%%%%%%%%%%%%

We will also present a preliminary result on NLQC classes not specified by a choice of Boolean function. 
In particular, inspired by the result of \cite{aaronson2020hardness}, we look for a relationship between the following two classes. 

\begin{definition}
    The \textbf{Distinguish} task for two orthogonal states $\ket{\phi_0}_{AB}, \ket{\phi_1}_{AB}$ is defined as follows. System $A$ is given to Alice and system $B$ is given to Bob; system $AB$ is in state $\ket{\phi_0}_{AB}$ with probability $1/2$, or state  $\ket{\phi_1}_{AB}$ with probability $1/2$.
    The goal is to output a bit $b$ labelling whether they were given the state $\ket{\phi_0}$ or $\ket{\phi_1}$.
    We say Distinguish is implemented $\epsilon$-correctly if the probability or outputting the correct bit $b$ is at least $1-\epsilon$. 
\end{definition}

\begin{definition}
    The \textbf{Interchange} task for two orthogonal states $\ket{\psi_0}_{AB}, \ket{\psi_1}_{AB}$ is defined as follows. 
    System $A$ is given to Alice and system $B$ is given to Bob; their goal is to map an arbitrary input superposition 
    $\alpha \ket{\psi_0} + \beta \ket{\psi_1}$ to the output superposition $\alpha \ket{\psi_1} + \beta \ket{\psi_0}$ (i.e., interchange $\ket{\psi_0}$ with $\ket{\psi_1}$).
\end{definition} 

%%%%%%%%%%%%%%%%%%%%%%%%%%%%%%%%%%%%%%%%%%%%%%%%%%%%%%%%%%%%%%%%%%%%%%
\section{Reductions among NLQCs}
%%%%%%%%%%%%%%%%%%%%%%%%%%%%%%%%%%%%%%%%%%%%%%%%%%%%%%%%%%%%%%%%%%%%%%

%%%%%%%%%%%%%%%%%%%%%%%%%%%%%%%%%%%%%%%%%%%%%%%%%%%%%%%%%%%%%%%%%%%%%%
\subsection{\texorpdfstring{$f$}{TEXT}-route, \texorpdfstring{$f$}{TEXT}-measure, and CDQS}
%%%%%%%%%%%%%%%%%%%%%%%%%%%%%%%%%%%%%%%%%%%%%%%%%%%%%%%%%%%%%%%%%%%%%%

The main result of this section is the following theorem.

\vspace{0.2cm}
\noindent \textbf{\Cref{thm:measureandroute}:} \emph{For any choice of Boolean function family, $f$-route, CDQS and $f$-measure are all equivalent under $O(1)$ resource state reductions.}
\vspace{0.2cm}

\noindent We will prove this by showing CDQS $\Rightarrow$ $f$-route $\Rightarrow$ $f$-measure $\Rightarrow$ CDQS, with each implication given as a separate lemma. 

The first lemma is already established in \cite{allerstorfer2024relating}. 
\begin{lemma} An $\epsilon$-correct and $\delta$-secure CDQS protocol for function $f$, hiding secret $Q$, and using resource system $\Psi_{LR}$ implies the existence of a $\max\{\epsilon,2\sqrt{\delta}\}$-correct $f$-route protocol that routes system $Q$ using resource system $\ket{\Psi}_{ELR}$ which purifies $\Psi_{LR}$, with $EL$ held by Alice and $R$ held by Bob. 
\end{lemma}

Note that this lemma gives an $f$-route protocol using the purification of the resource system $\Psi_{LR}$ used in the CDQS protocol. 
Recall that in our definition of a reduction, we always took the resource system to be pure, so that, in our definition, when considering the CDQS protocol we are starting out with a pure state resource and then using the same pure state resource in the $f$-route protocol. 
In some contexts it may be important to keep track of the fact that the CDQS protocol potentially can get away with a weaker resource (the unpurified resource system), but our definition of a reduction is too coarse to capture this distinction. 

Next, we move on to the implication from $f$-route to $f$-measure. 
\begin{lemma} \label{lem:frouteTOfmeasure}
    \textbf{$f$-route $\Rightarrow$ $f$-measure:} An $\epsilon$-correct $f$-route protocol for function $f$ using resource system $\Psi$ implies the existence of a $2\epsilon$-correct $f$-measure protocol for the same function using resource system $\Psi^{\otimes 2}$.
\end{lemma}
\begin{proof} We first describe the proof when the $f$-route protocol is perfectly correct; we then comment on the $\epsilon>0$ case. 

In the $f$-measure protocol, we are given input $x\in \{0,1\}^n$, $H^{f(x,y)}\ket{b}_A$ on Alice's side, $y\in\{0,1\}^n$ on Bob's side. 
Alice's protocol begins by copying system $A$ in the computational basis into a register $C$, then copying $A$ in the Hadamard basis into $B$, and $C$ in the Hadamard basis into $D$.  
Thus for example when $f(x,y)=0,b=0$, system $A$ is in state $\ket{0}$ and this procedure is
\begin{align}
    \ket{0}_A\overset{\text{copy}}{\rightarrow} \ket{00}_{AC} &=(\ket{+}_A+\ket{-}_A)(\ket{+}_C+\ket{-}_C) \nonumber \\
    &=\ket{++}_{AC}+ \ket{+-}_{AC}+ \ket{-+}_{AC} + \ket{--}_{AC} \nonumber \\
    &\overset{\text{copy}\, H}{\rightarrow} \ket{++++}_{ABCD}+ \ket{++--}_{ABCD}+ \ket{--++}_{ABCD}+ \ket{----}_{ABCD} \nonumber\\
    &=\ket{0000}_{ABCD}+ \ket{0011}_{ABCD}+ \ket{1100}_{ABCD}+\ket{1111}_{ABCD}, \nonumber 
\end{align}
where we haven't kept track of normalization. 
Repeating this for all four input states, we find
\begin{align}
    \ket{\Psi_{f=0,b=0}}_{ABCD} &= \frac{1}{2}\left(\ket{0000}+ \ket{0011}+ \ket{1100}+\ket{1111} \right), \nonumber \\
    \ket{\Psi_{f=0,b=1}}_{ABCD} &= \frac{1}{2}\left(\ket{0101}+ \ket{0110} + \ket{1001} + \ket{1010} \right), \nonumber \\
    \ket{\Psi_{f=1,b=0}}_{ABCD} &= \frac{1}{\sqrt{2}}\left(\ket{++++}+\ket{----} \right), \nonumber \\
    \ket{\Psi_{f=1,b=1}}_{ABCD} &=\frac{1}{\sqrt{2}}\left(\ket{++--}+\ket{--++} \right).
\end{align}
Alice and Bob perform $f$-route on the $B$ system, and $\neg f$-route on the $C$ system.
This uses two copies of the resource system for $f$-route.\footnote{Note that $f$-route and $\neg f$-route can be performed using the same resource system. This is because Alice and Bob can run the $f$-route protocol, then swap the message systems sent left and right to negate the function.} 
Alice always sends $D$ to Bob. 
Afterwards, Alice labels the two systems she holds (which may be $AB$ if $f(x,y)=0$ or $AC$ if $f(x,y)=1$) as $AB$, and Bob labels the two systems he holds ($CD$ if $f(x,y)=0$ or $BD$ if $f(x,y)=1$) as $CD$.
This produces the four states
\begin{align}\label{eq:postSWAP}
    \ket{\Psi_{f=0,b=0}}_{ABCD} &= \frac{1}{2}\left(\ket{0000}+ \ket{0011}+ \ket{1100}+\ket{1111} \right), \nonumber \\
    \ket{\Psi_{f=0,b=1}}_{ABCD} &= \frac{1}{2}\left(\ket{0101}+ \ket{0110} + \ket{1001} + \ket{1010} \right), \nonumber \\
    \ket{\Psi_{f=1,b=0}}_{ABCD} &= \frac{1}{\sqrt{2}}\left(\ket{++++}+\ket{----} \right), \nonumber \\
    \ket{\Psi_{f=1,b=1}}_{ABCD} &=\frac{1}{\sqrt{2}}\left(\ket{+-+-}+\ket{-+-+} \right).
\end{align}
In the second round, Alice and Bob know $x,y$ and hence know $f(x,y)$.
Then, Alice measures $AB$ and Bob measures $CD$; they measure each qubit in the computational basis if $f(x,y)=0$ and in the Hadamard basis if $f(x,y)=1$. 
Doing so, we see by inspection of the above states that Alice's two measurement outcomes will have even parity if $b=0$, and odd parity if $b=1$. 
The same is true of Bob's two measurement outcomes. 
Thus Alice and Bob both output the parity of their measurement outcomes, and this will correctly complete the $f$-measure task. 

Since the protocol uses oracle instances of protocols for $f$-route and $\neg f$-route, and the $\neg f$-route protocol is $\epsilon$ correct when the $f$-route protocol is, the overall protocol will be $2\epsilon$ correct by \Cref{remark:blackboxerror}.
\end{proof}

The final lemma needed is the implication from $f$-measure to CDQS'. 
For this result, we need to introduce some machinery developed in earlier analyses of $f$-measure \cite{bluhm2022single, asadi2025linear}. 
We follow the notation of \cite{asadi2025linear}.
We begin with the following definition. 

\begin{definition}\label{def:01setBB84}
    For $\epsilon \in [0,1]$, let $S^{\epsilon}_{C}$ be the set of states $\ket{\phi}_{PMM'}$ such that there exists a measurement on subsystem $M$ and a measurement on subsystem $M'$ that each determine the outcome of a measurement in the computational basis on $P$ with probability at least $1-\epsilon$. 
    Similarly, let $S^{\epsilon}_{H}$ be the set of states $\ket{\phi}_{PMM'}$ such that there exists a measurement on subsystem $M$ and a measurement on subsystem $M'$ that allows us to guess the outcome of a measurement in the Hadamard basis on $P$ with probability at least $1-\epsilon$.   
\end{definition}
Next we record the following lemma, proven in \cite{asadi2025linear}. 
\begin{lemma}\label{lemma:condentropybounds}
    Let $h(x)=-x\log x - (1-x) \log (1-x)$ be the binary entropy function.
    Suppose $\ket{\psi^C}_{PMM'} \in S^\epsilon_{C}$, and $\rho^C_{ZMM'}$ is obtained by measuring $P$ in the computational basis and recording the outcome in system $Z$.
    Then
    \begin{align*}
        S(Z|M')_{\rho^C_{ZMM'}} &\leq h(\epsilon)\, .
    \end{align*}
\end{lemma}

Now we are ready to prove that $f$-measure implies CDQS.
To show this, we actually first show $f$-measure implies CDQS' (where the secret is taken to be classical).
The intuition for the reduction is that an $f$-measure protocol has the effect of copying the input in the computational basis when $f(x,y)=0$, and in the Hadamard basis when $f(x,y)=1$. 
We exploit this functionality by preparing $H\ket{s}_Q$ and taking $Q$ as the quantum input to the $f$-measure protocol, so that the secret is copied correctly (and sent to the referee) when $f(x,y)=1$, but not when $f(x,y)=0$. 
Finally, \Cref{thm:CDQS'toCDQS} (proven in \Cref{sec:CDQSproperties}) shows that CDQS' implies CDQS, completing the chain of equivalences.

\begin{lemma} \label{lem:fmeasureTOcdqs} \textbf{$f$-measure $\Rightarrow$ CDQS':} Given a protocol for $f$-measure for function $f$ using resource system $\Psi$, there is a protocol for CDQS' for the same function using the same resource system. 
Further, as the error in the $f$-measure protocol goes to zero the correctness and security errors in the CDQS' protocol goes to zero.
\end{lemma}

\begin{proof}
    Beginning with the $f$-measure protocol, take the input system $Q$ to be in the state $H\ket{r}_Q$, with $r\in\{0,1\}$ chosen uniformly at random.
    The $f$-measure protocol produces system $M'$ which ends up on Alice's side after the communication round, and system $M$ which ends up on Bob's side. 
    In the CDQS' protocol, Alice and Bob run the same operations but then throw away $M'$ and send $M$ to the referee. 
    Alice additionally sends the bit $s\oplus r$ for $s$ the secret input to the CDQS. 
    We claim this defines a CDQS' protocol which hides $s$ and is both correct and secure. 
    
    Notice that the protocol involves running the first round operations of the $f$-measure protocol once, so uses one copy of the resource system $\Psi$. 
    We give the proof below assuming a single bit secret, but it extends easily to any $O(1)$ size secret. 

    \vspace{0.2cm}
    \noindent \textbf{Correctness:} By $\epsilon$-correctness of the $f$-measure protocol, when $f(x,y)=1$ Alice and Bob in the $f$-measure protocol can obtain $r$ with probability at least $1-\epsilon$.
    Since here the referee in the CDQS receives the systems Bob would have received, he can also determine $r$ with at least probability $1-\epsilon$, and hence recover $s$ from $s\oplus r$ and $r$ with the same probability.
    This means the CDQS' protocol is also $\epsilon$-correct. 

    \vspace{0.2cm}
    \noindent \textbf{Security:} In the case where $f(x,y)=0$, the intuition is that correctness of the $f$-measure protocol means a state in the computational basis inserted into the protocol would allow $r$ to be determined. 
    However, since we have recorded $r$ in the Hadamard basis, we expect $r$ (and hence $s$) cannot be determined perfectly. 
    To make this precise, we imagine preparing the system $Q$ in the state $H\ket{r}_Q$ by first preparing the maximally entangled state on $PQ$, then measuring $P$ in the Hadamard basis. 
    This measurement on $P$ commutes with the actions of Alice and Bob in the $f$-measure protocol, so we can first analyze that protocol as if $PQ$ were maximally entangled. 
    Recalling \Cref{def:01setBB84}, we have by $\epsilon$-correctness of the $f$-measure protocol that
    \begin{align}
        \rho_{PMM'} \in S^\epsilon_C.
    \end{align}
    Then, from \Cref{lemma:condentropybounds}, we have that
    \begin{align}
        S(Z|M')_{\rho^C_{ZMM'}} \leq h(\epsilon).
    \end{align}
    Then applying CIT we have that
    \begin{align}
        S(Z|M)_{\rho^H_{ZMM'}} \geq 1- S(Z|M')_{\rho^{C}_{ZMM'}} \geq 1-h(\epsilon),
    \end{align}
    which says that, from the referee's perspective, $Z$ (which determines $r$) is high entropy, suggesting the referee cannot learn $s$.

    To prove security of the CDQS' protocol formally, we need that the state the referee sees is close to one that doesn't depend on $s$.  
    We will first show that a protocol where Alice and Bob only send $M$ is secure (in the sense that $\rho_{M}$ is independent of $r$).
    To show this, note that the general form of $\rho^H_{ZM}$ is
    \begin{align}\label{eq:form}
        \rho^H_{ZM}(r,x,y) = \frac{1}{2}\sum_r \ketbra{r}{r}_Z\otimes \rho_{M}'(r,x,y).
    \end{align}
    We now use
    \begin{align}
        I(Z:M)_{\rho^H} = S(Z)_{\rho^H} - S(Z|M)_{\rho^H}  
    \end{align}
    and the lower bound above on $S(Z|M)_{\rho^H}$ to bound
    \begin{align}
        I(Z:M)_{\rho^H} \leq h(\epsilon).
    \end{align}
    Next use $I(A:B)_\rho=D(\rho_{AB}\|\rho_A\otimes \rho_{B})$ and Pinsker's inequality \eqref{eq:pinsker} to obtain
    \begin{align}\label{eq:TDH}
        \left\Vert \frac{1}{2}\sum_r \ketbra{r}{r}\otimes \rho^H_M(x,y) - \frac{1}{2}\sum_r \ketbra{r}{r}_Z\otimes \rho'_{M}(r,x,y) \right\Vert_1 \leq \sqrt{(2\ln 2)h(\epsilon)},
    \end{align}
    where $\rho^H_M$ is formed by tracing out $Z$ from $\rho^H_{ZM}$, and is from \Cref{eq:form}
    \begin{align}
        \rho^H_{M}(x,y) = \frac{1}{2}\sum_r \rho'_M(r,x,y).
    \end{align}
    Then, pulling the sum out of the trace distance in \Cref{eq:TDH}, we obtain
    \begin{align}
         \frac{1}{2}\sum_r \left\Vert \rho^H_{M}(x,y) - \rho'_{M}(r,x,y) \right\Vert_1 \leq \sqrt{(2\ln 2)h(\epsilon)}.
    \end{align}
    From this we have that, for all $r$, we must have
    \begin{align}
        \left\Vert \rho^H_{M}(x,y) - \rho_{M}'(r,x,y) \right\Vert_1 \leq 2\sqrt{(2\ln 2)h(\epsilon)}.
    \end{align}
    Comparing to the security \Cref{def:CDQSprime} we see that at small enough $\epsilon$, the CDQS' protocol becomes arbitrarily secure (where the secret is $r$), as needed. 
\end{proof}

\begin{remark}\textbf{$f$-measure $\Rightarrow$ CDQS:}
    By \Cref{thm:CDQS'toCDQS} along with \Cref{lem:fmeasureTOcdqs} we have that we can construct a CDQS protocol using $O(1)$ copies of the $f$-measure resource state $\Psi$.
\end{remark}

%%%%%%%%%%%%%%%%%%%%%%%%%%%%%%%%%%%%%%%%%%%%%%%%%%%%%%%%%%%%%%%%%%%%%%
\subsection{Coherent unitary and coherent phase}
%%%%%%%%%%%%%%%%%%%%%%%%%%%%%%%%%%%%%%%%%%%%%%%%%%%%%%%%%%%%%%%%%%%%%%

In this section we prove the implications $Cf$-SWAP $\Rightarrow$ $Cf$-PHASE $\Rightarrow$ $Cf$-PAULI. 

The coherently controlled SWAP NLQC asks us to implement the unitary
\begin{align}
    Cf\text{-SWAP}_{AA'BB'} = \sum_{x,y} \textnormal{SWAP}_{A'B'}^{f(x,y)}\otimes \ketbra{x}{x}_A\otimes \ketbra{y}{y}_B\,.
\end{align}
We can see that a protocol for $Cf$-SWAP implies a protocol for $Cf$-PHASE, which recall asks us to implement the unitary
\begin{align}
    Cf\text{-PHASE}_{AB} = \sum_{x,y} (-1)^{f(x,y)} \ketbra{x}{x}_A\otimes\ketbra{y}{y}_B.
\end{align}
We state this implication as the next lemma. 

\begin{lemma} \label{lem:cfswapTOcfphase}
    For every family of Boolean functions, there is a $O(1)$ implication from $Cf$-SWAP to $Cf$-PHASE. 
\end{lemma}
\begin{proof}
First consider the case where we have access to an exactly correct $Cf$-SWAP protocol. 
Alice and Bob hold one extra EPR pair, which we take to be on the $A'B'$ systems, $\Psi^+_{A'B'}$. 
Have Alice apply $\mathbf{Z}_{A'}\mathbf{X}_{A'}$ operator to put the $A'B'$ systems in the state
\begin{align}
    \ket{\Phi^-} = \frac{1}{\sqrt{2}}(\ket{01} - \ket{10}).
\end{align}
Then, run the $Cf$-SWAP protocol with this state as the input on $A'B'$, and then trace out $A'B'$ at the end, which will implement the operation
\begin{align}
    \ket{\Phi^-}_{A'B'}\ket{x}_A\ket{y}_B \overset{Cf\text{-SWAP}}{\rightarrow}& \text{SWAP}^{f(x,y)}_{A'B'}\ket{\Phi^-}_{A'B'} \ket{x}_A\ket{y}_B \nonumber \\
    =\,\,\,\,\,\,\,\,& (-1)^{f(x,y)}\ket{\Phi^-}_{A'B'} \ket{x}_A\ket{y}_B \nonumber \\
    \overset{\tr_{A'B'}}{\rightarrow}\,\,\,\,& (-1)^{f(x,y)}\ket{x}_A\ket{y}_B\,.
\end{align}
This is exactly the action of the $Cf$-PHASE operator on a basis, as needed. 

Observe that this protocol uses the $Cf$-SWAP protocol as an oracle, so \Cref{remark:blackboxerror} implies that $Cf$-SWAP protocol with error $\epsilon$ gives a $Cf$-PHASE protocol with error at most $\epsilon$. 
\end{proof}

Note that we have not been able to show an implication from $Cf$-PHASE to $Cf$-SWAP, but we do not have any formal evidence this is not possible. 

Continuing, we can find an implication from $Cf$-PHASE to $Cf$-PAULI, which choosing the Pauli to be $\mathbf{Z}$ is the unitary
\begin{align}
    Cf\text{-}\mathbf{Z}_{A'AB} = \sum_{x,y} \mathbf{Z}_{A'}^{f(x,y)} \otimes \ketbra{x}{x}_A\otimes \ketbra{y}{y}_B.
\end{align}
We can now state the following theorem. 
\begin{theorem} \label{thm:cfphaseTOcfpauli}
    For every family of Boolean functions, there is a $O(1)$ implication from $Cf$-PHASE to $Cf$-PAULI and from $Cf$-PAULI to $Cf$-PHASE. 
\end{theorem}
\begin{proof}
Consider the action of $Cf$-$\mathbf{Z}$ on an arbitrary basis state, $\ket{z}_{A'}\ket{x}_A\ket{y}_B$,
\begin{align}
    Cf\text{-}\mathbf{Z}_{A'AB}\ket{z}_{A'}\ket{x}_A\ket{y}_B = \mathbf{Z}^{f(x,y)}_{A'}\ket{z}_{A'}\ket{x}_A\ket{y}_B = (-1)^{f(x,y)\wedge z}\ket{z}_{A'}\ket{x}_A\ket{y}_B.
\end{align}
Thus, a $Cf$-PHASE protocol for the function $f(x,y)\wedge z$ is exactly a $Cf$-$\mathbf{Z}$ protocol, where the input qubit for the $Cf$-PHASE that holds $z$ is treated as the target qubit in the $Cf$-$\mathbf{Z}$ protocol. 
Let the resource system used in the $Cf$-PHASE protocol for function $f$ be $\Psi$.
Using \Cref{thm:PhaseANDclosure}, we know the $Cf$-PHASE protocol for $f(x,y)\wedge z$ can be implemented using $\Psi\otimes \Psi^+$, and so $Cf$-$\mathbf{Z}$ can be as well. 

To obtain $Cf$-PHASE from $Cf$-PAULI, we choose $A'$ to be in the state $\ket{1}_{A'}$ and run the $Cf$-PAULI protocol. 

To handle the approximate case, we again observe that because the above reductions are oracle reduction we can employ \Cref{remark:blackboxerror}.
\end{proof}

$Cf$-PHASE can also be used to implement other coherently controlled unitaries. 
In particular any unitaries that act on qubits held entirely by Alice (or entirely by Bob) and have $\pm 1$ eigenvalues can be implemented via a similar strategy to the one used to implement $Cf$-$\mathbf{Z}$. 
To implement $\mathbf{U}_{A'}$ Alice should first apply a unitary $\mathbf{V}_{A'}$ to map the inputs from the eigenbasis of $\mathbf{U}_{A'}$ to the computational basis, implement a controlled phase as needed, then apply $\mathbf{V}^\dagger$. 
For instance, a $CZ$ gate can be implemented in this way with $\mathbf{V}_{A'}=\mathcal{I}_{A'}$ and using the controlled phase operation $f(x,y)\wedge z_1 \wedge z_2$. 
Similarly, $\text{SWAP}_{A'_1A'2}$ can be implemented by choosing $\mathbf{V}_{A'_1A'_2}$ to map from the Bell basis to the computational basis and a controlled phase with function $f(x,y)\wedge z_1 \wedge z_2$ again.\footnote{Note that while coherently controlled $\text{SWAP}_{A'_1A'_2}$ follows from $Cf$-PHASE, this does not imply that  controlled $\text{SWAP}_{A'B'}$ can, where in the later case the two qubits being swapped are separately on Alice and Bob's side. In this article when we refer to $Cf$-SWAP we mean implementing $\text{SWAP}_{A'B'}$.} 
Similar strategies to those described here also relate controlled phase operations with other choices of phase $\omega \neq \pm 1$ and unitaries with eigenvalues $1,\omega$. 

%%%%%%%%%%%%%%%%%%%%%%%%%%%%%%%%%%%%%%%%%%%%%%%%%%%%%%%%%%%%%%%%%%%%%%
\subsection{Coherent Pauli to CDQS}
%%%%%%%%%%%%%%%%%%%%%%%%%%%%%%%%%%%%%%%%%%%%%%%%%%%%%%%%%%%%%%%%%%%%%%

Previously, we noted that $Cf$-SWAP $\Rightarrow$ $f$-route, and hence $Cf$-SWAP also implies the other classically controlled protocols we study. 
Here, we show that in fact the potentially weaker coherent primitive of $Cf$-PAULI already suffices to obtain the classically controlled classes.

\begin{theorem}\label{thm:PaulitoCDQS}
    For every family of Boolean functions, there is a $O(1)$ implication from $Cf$-PAULI to CDQS. 
\end{theorem}

\begin{proof}
    We use the $Cf$-PAULI protocol to implement a CDQS' protocol. \Cref{thm:CDQS'toCDQS} then gives us the CDQS protocol. We use the following procedure. 
    \begin{enumerate}
        \item From their classical inputs Alice prepares quantum systems $\ket{s}_{A'}\ket{x}_A$; Bob prepares $\ket{0}_{B'}\ket{y}_B$.
        \item Alice and Bob execute the first round operations of a protocol for $Cf$-$\mathbf{X}_{B'}$ with choice of function $f(x,y)\wedge s$. 
        \item The systems that would be sent to Bob (call them system $M$) in the second round are sent to the referee in the CDQS protocol. The remaining degrees of freedom which would be sent to Alice (call them system $M'$) are traced out.
    \end{enumerate}
    Using \Cref{thm:PhaseANDclosure} and the equivalence of $Cf$-PAULI and $Cf$-PHASE, we have that this uses one copy of the resource system for a $Cf$-$\mathbf{X}$ protocol for function $f$, plus $O(1)$ additional EPR pairs.
    
    \vspace{0.2cm}
    \noindent \textbf{Correctness:} We have that $f(x,y)=1$. Let the action of (both rounds of) the $Cf$-$\mathbf{X}$ protocol be described by a channel $\mathbfcal{N}$. 
    Then, perfect correctness of the $Cf$-$\mathbf{X}$ protocol would give that 
    \begin{align}
        \mathbfcal{N}\left(\ketbra{s,x}{s,x}_{A'A}\otimes \ketbra{0,y}{0,y}_{B'B}\right) &= \ketbra{s,x}{s,x}_{A'A}\otimes \ketbra{f(x,y)\wedge s,y}{f(x,y)\wedge s,y}_{B'B} \nonumber \\
        &= \ketbra{s,x}{s,x}_{A'A}\otimes \ketbra{s,y}{s,y}_{B'B}, \nonumber 
    \end{align}
    where in the second line we used that $f(x,y)=1$.
    Since the referee in the CDQS gets the same systems as Bob obtains in the $Cf$-$\mathbf{X}$, the referee can apply the same local channel as Bob and hence recover the $B'B$ portion of the above density matrix. 
    Thus the referee in the perfectly correct case can measure $B'$ and learn $f(x,y)\wedge s=s$.
    
    To obtain approximate correctness when the $Cf$-$\mathbf{X}$ protocol is approximately correct, we note that the CDQS protocol we implement produces the same state on $B'B$ as if we implemented the $Cf$-$\mathbf{X}$ protocol as an oracle, and then traced out the system on the left. 
    Because of this, \Cref{remark:blackboxerror} implies the protocol here is $\epsilon$-correct if the $Cf$-$\mathbf{X}$ protocol is $\epsilon$ correct. 

    \vspace{0.2cm}
    \noindent \textbf{Security:} Let $\mathbf{V}_{SAB'B\rightarrow MM'}$ be an isometric extension of Alice and Bob's first round operations in the $Cf-\mathbf{X}$ protocol. 
    Let $\mathbf{W}^L_{M'\rightarrow SAl}$ and $\mathbf{W}^R_{M\rightarrow B'Br}$ to be isometric extensions of the second round operations, with $l$ and $r$ the purifying systems. 
    Define channels
    \begin{align}   
        \mathbfcal{V}_{SAB'B\rightarrow MM'}&= \mathbf{V}_{SAB'B\rightarrow MM'}(\cdot) \mathbf{V}_{SAB'B\rightarrow MM'}^\dagger\nonumber \\
        \mathbfcal{W}^L_{M'\rightarrow A'Al} &= \mathbf{W}^L_{M'\rightarrow A'Al}(\cdot)(\mathbf{W}^L_{M'\rightarrow A'Al})^\dagger \nonumber \\
        \mathbfcal{W}^R_{M'\rightarrow A'Ar} &= \mathbf{W}^R_{M\rightarrow B'Br}(\cdot)(\mathbf{W}^R_{M\rightarrow B'Br})^\dagger.
    \end{align}
    Then the action of the $Cf$-$\mathbf{X}$ protocol, in $f(x,y)=0$ instances and assuming perfect correctness of the $Cf$-$\mathbf{X}$, is
    \begin{align}
        \ketbra{s,x}{s,x}_{A'A}&\otimes \ketbra{0,y}{0,y}_{B'B} \nonumber \\
        &= \tr_{lr}(\mathbfcal{W}_{M'\rightarrow A'Al}^L\otimes \mathbfcal{W}_{M\rightarrow B'Br}^R) \circ \mathbfcal{V}_{A'AB'B\rightarrow M'M}(\ketbra{s,x}{s,x}_{A'}\otimes \ketbra{0,y}{0,y}_{B'}). \nonumber 
    \end{align}
    Removing the trace, we have
    \begin{align}\label{eq:step}
        \ket{s,x}_{A'A} &\otimes \ket{0,y}_{B'B} \otimes \ket{\Psi(s,x,y)}_{lr} \nonumber \\
        &= (\mathbf{W}_{M'\rightarrow A'Al}^L\otimes \mathbf{W}_{M\rightarrow B'Br}^R) \circ \mathbf{V}_{A'AB'B\rightarrow M'M}\ket{s,x}_{A'}\otimes \ket{0,y}_{B'}) \nonumber \\
        &=\mathbf{M}_{A'AB'B\rightarrow A'AB'Blr} \ket{s,x,y},
    \end{align}
    where the last line defines $\mathbf{M}$.
    We can actually simplify the above by noting that the state on $lr$ must be independent of $s,x,y$. 
    To see why, suppose we had inserted a superposition of $\ket{s,x,y}$ basis states, producing
    \begin{align}
        \sum_{s,x,y} c_{s,x,y}  \ket{s,x,0,y} \rightarrow \sum_{s,x,y} c_{s,x,y} \mathbf{M}\ket{s,x,0,y} = \sum_{s,x,y} c_{s,x,y} \ket{s,x,0,y}_{AA'BB'}\ket{\Psi(s,x,y)}_{lr}.
    \end{align}
    But now tracing out $lr$ we obtain
    \begin{align}
        \sum_{s,x,y} c_{s,x,y}\bar{c}_{s',x',y'} \braket{\Psi(s,x,y)}{\Psi(s',x',y')} \ketbra{s,x,0,y}{s',x',0,y'}_{A'AB'B} .
    \end{align}
    By correctness, this must equal the output of the correct protocol, so that
    \begin{align}
        \sum_{s,x,y} c_{s,x,y}\bar{c}_{s',x',y'}& \ketbra{s,x,0,y}{s',x',0,y'}_{A'AB'B}  \nonumber \\
        &= \sum_{s,x,y} c_{s,x,y}\bar{c}_{s',x',y'} \braket{\Psi(s,x,y)}{\Psi(s',x',y')} \ketbra{s,x,0,y}{s',x',0,y'}_{A'AB'B}.\nonumber 
    \end{align}
    This gives that $\braket{\Psi(s,x,y)}{\Psi(s',x',y')}=1$, so that $\ket{\Psi(s,x,y}$ must not depend on $s,x,y$.
    Using this, we can simplify \Cref{eq:step} to take $\Psi_{lr}$ to be independent of $s,x,y$, so that
    \begin{align}
        (\ketbra{s,x}{s,x}_{A'}&\otimes \ketbra{0,y}{0,y}_{B'B}\otimes \Psi_{lr}) \nonumber \\
        &= (\mathbfcal{W}^L_{M'\rightarrow A'A} \otimes \mathbfcal{W}^R_{M\rightarrow B'B})\circ \mathbfcal{V}_{A'AB'B\rightarrow M'M}(\ketbra{s,x}{s,x}_{A'}\otimes \ketbra{0,y}{0,y}_{B'}). \nonumber 
    \end{align}
    Next act with $(\mathbfcal{W}_{M'\rightarrow A'Al}^L\otimes \mathbfcal{W}^R_{M\rightarrow B'Br})^\dagger$ to produce
    \begin{align}
        (\mathbfcal{W}^L_{M'\rightarrow A'A})^\dagger \otimes (\mathbfcal{W}^R_{M\rightarrow B'B})^\dagger (\ketbra{s,x}{s,x}_{A'}\otimes &\ketbra{0,y}{0,y}_{B'B}\otimes \Psi_{lr}) \nonumber \\
        &= \mathbfcal{V}_{A'AB'B\rightarrow M'M}(\ketbra{s,x}{s,x}_{A'}\otimes \ketbra{0,y}{0,y}_{B'}) \nonumber \\
        &= \rho_{M'M}(x,y,s).
    \end{align}
    Tracing out $M'$ to find the density matrix that describes $M$, we see that the system the referee obtains is
    \begin{align}
        \rho_M(x,y,s) &= \tr_{M'} (\mathbfcal{W}^L_{M'\rightarrow A'Al})^\dagger \otimes (\mathbfcal{W}^R_{M\rightarrow B'Br})^\dagger (\ketbra{s,x}{s,x}_{A'A}\otimes \ketbra{0,y}{0,y}_{B'B}\otimes \Psi_{lr}) \nonumber \\
        &= \tr_{A'Al}  (\mathbfcal{W}^R_{M\rightarrow B'Br})^\dagger (\ketbra{s,x}{s,x}_{A'Al}\otimes \ketbra{0,y}{0,y}_{B'Br}\otimes \Psi_{lr}) \nonumber \\
        &= (\mathbfcal{W}^R_{M\rightarrow B'Br})^\dagger (\ketbra{0,y}{0,y}_{B'Br}\otimes \Psi_r) \nonumber \\
        &=: \sigma_M(y).
    \end{align}
    Since this state does not depend on $s$, the CDQS protocol is perfectly secure. 

    To establish approximate security in the case where $Cf$-$\mathbf{X}$ is approximately correct, we again note that the state produced on the referee's side is the same one as if we had made an oracle invocation of the $Cf$-$\mathbf{X}$ protocol. Thus the state produced must be $\epsilon$-close to $\sigma(y)$, which is exactly $\epsilon$ security of the CDQS protocol. 
\end{proof}

%%%%%%%%%%%%%%%%%%%%%%%%%%%%%%%%%%%%%%%%%%%%%%%%%%%%%%%%%%%%%%%%%%%%%%
\subsection{Interchange implies Distinguish}
%%%%%%%%%%%%%%%%%%%%%%%%%%%%%%%%%%%%%%%%%%%%%%%%%%%%%%%%%%%%%%%%%%%%%%

In this section we give a different type of reduction, this one between two NLQC tasks that aren't characterized in terms of a Boolean function. 
These tasks are Interchange (i.e., coherently transform a state $\ket{\psi_0}$ to an orthogonal state $\ket{\psi_1}$ and vice versa), and Distinguish (i.e., distinguish between two orthogonal states $\ket{\phi_0},\ket{\phi_1}$). See \Cref{sec:stateandunitary} for formal definitions of these tasks. 
At a high level, we show that an efficient NLQC protocol for Interchange with states $\ket{\psi_0},\ket{\psi_1}$ implies an efficient NLQC protocol for Distinguish with the ``Fourier'' states $\ket{\phi_{\pm}} = \frac{1}{\sqrt{2}} (\ket{\psi_0} \pm \ket{\psi_1})$. 

Our result is inspired by the equivalence of these tasks in the standard model of quantum circuits: it is known~\cite{aaronson2020hardness} that an efficient algorithm for Interchange with $\ket{\psi_0},\ket{\psi_1}$ implies an efficient algorithm for Distinguish $\ket{\phi_\pm}$, \emph{and vice versa}. 
Note that in this setting we are considering global operations, with no division of the inputs between players. 
This equivalence, though simple, has had important applications in quantum cryptography~\cite{bostanci2023unitary,hhan2023hardness,morimae2025quantum} in recent years.

While we found a reduction in one direction (Interchange implies Distinguish) in the NLQC setting, we weren't able to find the other direction. 
It is an interesting open question of whether this other direction is possible. 

Before describing our reduction more precisely, we first define the notion of a \textbf{catalytic} NLQC protocol. 
Intuitively, it is one where the resource state is used as part of the protocol, and then returned to its original state at the end. 

\begin{definition}[Catalytic NLQC protocol]
Let $\mathbf{U}_{AB \rightarrow AB}$ denote a unitary $2\rightarrow 2$ task. We say that a NLQC $\mathbfcal{N}_{AB \to AB}$ for $\mathbf{U}$ is \textbf{catalytic} if the isometric extensions $V,W$ of the first and second round operations and the resource state $\Psi_{LR}$ satisfy the following: for all input states $\ket{\phi}_{AB}$,
\[
    \mathbf{W}\mathbf{V} \ket{\phi}_{AB} \otimes \ket{\Psi}_{LR} = (\mathbf{U} \ket{\phi}_{AB}) \otimes \ket{\Psi}_{LR}~.
\]
In other words, the resource state $\ket{\Psi}_{LR}$ serves as a \emph{catalyst} for the nonlocal computation. 
\end{definition}

We now argue that any NLQC for a unitary task can be made \emph{approximately} catalytic, with only a linear factor overhead in the size of the resource state.

\begin{lemma}
\label{lem:approx-catalysis}
    Let $\mathbf{U}$ denote a unitary $2\rightarrow 2$ task. Suppose that $\mathbfcal{N}_{AB \to AB}$ is a $\delta$-correct NLQC for $\mathbf{U}$ (not necessarily catalytic) with an $m$-qubit resource state. Then for all $\epsilon > 0$, there exists an $(\epsilon + \sqrt{\delta})$-approximately catalytic NLQC $\rule{0pt}{3ex}\hat{\mathbfcal{N}}_{AB \to AB}$ for $\mathbf{U}$ using a $O(m/\epsilon^2)$-qubit resource state $\Xi_{L'R'}$, in the sense that for all input states $\ket{\phi}_{AB}$,
    \[
        \| \hat{\mathbf{W}}\hat{\mathbf{V}} (\ket{\phi}_{AB} \otimes \ket{\Xi}_{L'R'}) - (\mathbf{U} \ket{\phi}_{AB}) \otimes \ket{\Xi}_{L'R'} \|_2 \leq \epsilon + \sqrt{\delta}
    \]
    where $\hat{\mathbf{V}},\hat{\mathbf{W}}$ are the isometric extensions of the first and second round operations, respectively, of the NLQC $\hat{\mathbfcal{N}}$.
\end{lemma}

\begin{proof}
This proof relies on the idea of \emph{entanglement embezzlement states}~\cite{van2003universal}; these are a family of bipartite states $\{ \ket{\Xi_n}_{LR} \}$ where $\ket{\Xi_n}$ is on $n$-qubits, such that for all bipartite pure states $\ket{\psi}_{AB}$ on $m$ qubits, there exists local isometries $\mathbf{S}_{L \to LA}, \mathbf{T}_{R \to RB}$ such that
\begin{align}
    \| \mathbf{S}_{L\rightarrow LA} \otimes \mathbf{T}_{R\rightarrow RB} \ket{\Xi_n}_{LR} - \ket{\Xi_n}_{LR} \otimes \ket{\psi}_{AB} \|_2^2 \leq \frac{2m}{n}~.
\end{align}
In other words, using only local operations, the entanglement $\ket{\psi}$ has been \emph{embezzled} from the state $\ket{\Xi_n}$, and this is not very noticeable (provided that $n \gg m$). 

From $\delta$-correctness of $\mathbfcal{N}_{AB\rightarrow AB}$, we have that for an arbitrary input state $\ket{\psi}_{AB}$
\begin{align}
    \|\mathbfcal{N}_{AB\rightarrow AB}(\ketbra{\psi}{\psi}_{AB}) - \mathbf{U}_{AB}\ketbra{\psi}{\psi}\mathbf{U}^\dagger_{AB} \|_1 \leq \delta.
\end{align}
Let $\mathbf{V},\mathbf{W}$ be isometric extensions of the first- and second-round operations of $\mathbfcal{N}$, respectively.
Let $\ket{\Psi}_{LR}$ denote the resource state, and $\tilde{L}\tilde{R}$ the Hilbert space of the purifying system at the end of the protocol. 
Then the above gives that
\begin{align}
    \|\tr_{\tilde{L}\tilde{R}}(\mathbf{W}\mathbf{V}(\ketbra{\psi}{\psi}\otimes \ketbra{\Psi}{\Psi})\mathbf{V}^\dagger \mathbf{W}^\dagger) - \mathbf{U}\ketbra{\psi}{\psi}\mathbf{U}^\dagger \|_1 \leq \delta.
\end{align}
Converting this to a fidelity using the Fuchs van de Graff inequalities, applying Uhlmann's theorem to express the fidelity in terms of purified states, and converting back to a trace distance we obtain that there exists a state $\ket{\Theta}_{\tilde{L}\tilde{R}}$ such that
\begin{align} \label{eq:after-Uhlmann}
    \|\mathbf{W}\mathbf{V}(\ketbra{\psi}\otimes \ketbra{\Psi})\mathbf{V}^\dagger\mathbf{W}^\dagger - \mathbf{U}\ketbra{\psi}\mathbf{U}^\dagger\otimes \ketbra{\Theta} \|_1 \leq \sqrt{2\delta} \, .
\end{align}
Using that $\|\ketbra{u} - \ketbra{v}\|_1 = 2 \sqrt{1-|\braket{u}{v}|^2}\geq 2\sqrt{1-|\braket{u}{v}|}$ and the fact that we are free to choose the phase of $\ket{\Theta}$ such that the inner product of the vectors in \eqref{eq:after-Uhlmann} is positive, we can also infer that
\begin{align}
    \|\mathbf{W}\mathbf{V}(\ket{\psi}\otimes \ket{\Psi}) - \mathbf{U}\ket{\psi}\otimes \ket{\Theta} \|_2 \leq \sqrt{\delta} \,.
\end{align}
%It must be that, there exists another state $\ket{\Theta}_{LR}$ such that for all input states $\ket{\phi}_{AB}$
%\[
%    \| WV (\ket{\phi} \otimes \ket{\Psi}_{LR}) - (\mathbf{U} \ket{\phi}_{AB}) \otimes \ket{\Theta}_{LR} \|_2 \leq \delta~.
%\]
In other words, in the process of transforming $\ket{\phi}$ to $\mathbf{U} \ket{\phi}$, the resource state $\ket{\Psi}$ has been transformed to another resource state $\ket{\Theta}$. 
This is not a catalytic protocol when $\ket{\Psi}$ is far from $\ket{\Theta}$. 

We now adapt this to be a catalytic NLQC. 
Consider the following NLQC. 
Set $n = 8m/\epsilon^2$, and consider the $n$-qubit embezzlement state $\ket{\Xi_n}_{L'R'}$. 
By the above there exist unitary operations $\mathbf{S}_{L'L \to L'L},\mathbf{T}_{R'R \to R'R}$ such that $\mathbf{S} \otimes \mathbf{T}$, acting on $\ket{\Xi_n}\otimes \ket{0}^{\otimes m}$, is $\epsilon/2$-close to $\ket{\Xi_n}_{L'R'} \otimes \ket{\Psi}_{LR}$. First, apply these local operations. Then, apply the original isometries $\mathbf{V},\mathbf{W}$ with the resource state $\ket{\Psi}$ as in $\mathbfcal{N}$, and finally apply the inverse of a pair of local unitaries $\mathbf{S}' \otimes \mathbf{T}'$ that ``reverse'' the embezzlement, which satisfy
\begin{align}
    \| \mathbf{S}' \otimes \mathbf{T}' \ket{\Xi_n} \otimes \ket{0 \cdots 0} - \ket{\Xi_n} \otimes \ket{\Theta} \|_2 \leq \epsilon/2~.
\end{align}
Define the isometric extensions $\hat{\mathbf{V}} = \mathbf{V} (\mathbf{S} \otimes \mathbf{T})$ and $\hat{\mathbf{W}} = ((\mathbf{S}')^\dagger \otimes (\mathbf{T}')^\dagger) \mathbf{W}$. Since $\mathbf{V},\mathbf{W}$ factor into local operations (i.e., $\mathbf{V} = \mathbf{V}_{AL} \otimes \mathbf{V}_{BR}$ and $\mathbf{W} = \mathbf{W}_{AL} \otimes \mathbf{W}_{BR}$) these isometries $\hat{\mathbf{V}},\hat{\mathbf{W}}$ yield a NLQC. 
By the triangle inequality this NLQC for $\mathbf{U}$ is $(\epsilon+\sqrt{\delta})$-approximately catalytic.
\end{proof}

Finally, we now show that any efficient, catalytic, Interchange protocol (for which \Cref{lem:approx-catalysis} shows the catalytic property is without loss of generality) implies an efficient Distinguish protocol for a related pair of states.  

\begin{theorem}
    Suppose there is an $\epsilon$-correct catalytic NLQC for Interchange between two orthogonal states $\ket{\psi_0}_{AB},\ket{\psi_1}_{AB}$ using an $m$-qubit resource state. 
    Then there exists a $\epsilon$ correct NLQC for Distinguish between the ``Fourier'' states $\ket{\phi_\pm} = \frac{1}{\sqrt{2}} (\ket{\psi_0} \pm \ket{\psi_1})$ that uses the same $m$-qubit resource state, with two extra EPR pairs. 
\end{theorem}

\begin{proof}
    Label an isometric extension of the first round operations in the catalytic Interchange protocol $\mathbf{V}$, and an isometric extension of the second round operations $\mathbf{W}$. Let the catalytic resource state used by the Interchange protocol be $\ket{\Phi}_{LR}$. 
    Define states
    \begin{align}
        \ket{\phi_0}&=\mathbf{V}\ket{\psi_0} \ket{\Phi}\nonumber \\
        \ket{\phi_1}&=\mathbf{V}\ket{\psi_1} \ket{\Phi}
    \end{align}
    and notice that since by correctness $\Vert \ket{\psi_1}\ket{\Phi}-\mathbf{W}\mathbf{V}\ket{\psi_0}\ket{\Phi}\Vert_1\leq \epsilon$ and $\Vert \ket {\psi_0} \ket{\Phi}-\mathbf{W}\mathbf{V}\ket{\psi_1} \ket{\Phi}\Vert_1\leq \epsilon$, we have also
    \begin{align}
        \Vert\ket{\psi_1} \ket{\Phi} - \mathbf{W}\ket{\phi_0}\Vert_1 &\leq \epsilon, \nonumber \\
        \Vert \ket{\psi_0} \ket{\Phi} - \mathbf{W}\ket{\phi_1}\Vert_1 &\leq \epsilon.
    \end{align}
    Now consider the following protocol for Distinguish. 
    Distribute one extra EPR pair, so that Alice and Bob hold
    \begin{align}
        \frac{1}{2}\left(\ket{\psi_0}_{AB} + (-1)^b \ket{\psi_1}_{AB} \right)\left( \ket{00}_{A_0B_0}+\ket{11}_{A_0B_0}\right) \ket{\Phi}_{LR}~.
    \end{align}
    Then, we have Alice and Bob do the following: controlled on $A_0$, Alice implements her operations from $\mathbf{V}$ if $A_0$ is in state $\ket{0}$, and her operations from $\mathbf{W}^\dagger$ if it is in state $\ket{1}$. 
    Similarly, Bob does his operations from $\mathbf{V}$ or $\mathbf{W}^\dagger$ controlled on $B_0$. 
    This brings the state to one $\epsilon$-close in trace distance to
    \begin{align}
        \frac{1}{2}\left(\ket{\phi_0}_{AB} + (-1)^b \ket{\phi_1}_{AB} \right)&\ket{00}_{A_0B_0}+\frac{1}{2}\left(\ket{\phi_1}_{AB} + (-1)^b \ket{\phi_0}_{AB} \right)\ket{11}_{A_0B_0} \nonumber \\
        &= \frac{1}{2}\left(\ket{\phi_0}_{AB} + (-1)^b \ket{\phi_1}_{AB} \right)\left( \ket{00}_{A_0B_0}+(-1)^b\ket{11}_{A_0B_0}\right) \,.\nonumber 
    \end{align}
    Next, Alice and Bob throw away the $AB$ systems. 
    This leaves the $A_0B_0$ systems in a state $\epsilon$-close to $\frac{1}{\sqrt{2}}\left( \ket{00}_{A_0B_0}+(-1)^b\ket{11}_{A_0B_0}\right)$ and run Distinguish on the inputs $\ket{00}, \ket{11}$. 
    This can be implemented using one EPR pair, and will be $\epsilon$ correct.  
    Along with the EPR pair used above, this implements Distinguish on the inputs $\ket{\psi_0}, \ket{\psi_1}$ with two extra EPR pairs compared to the catalytic Interchange protocol. 
\end{proof}

%%%%%%%%%%%%%%%%%%%%%%%%%%%%%%%%%%%%%%%%%%%%%%%%%%%%%%%%%%%%%%%%%%%%%%
\section{Implications and simplifications}\label{sec:simplifications}
%%%%%%%%%%%%%%%%%%%%%%%%%%%%%%%%%%%%%%%%%%%%%%%%%%%%%%%%%%%%%%%%%%%%%%

Because $f$-route and $f$-measure are well studied in the literature, our proof that they are equivalent leads to a large number of simplifications. 
In particular, the papers \cite{asadi2024rank, asadi2025linear, bluhm2022single, escola2025quantum} each contain repeated proofs, once for $f$-measure and once for $f$-route, which can now become one proof for either measure or routing, and then \Cref{thm:measureandroute} ensures the other NLQC inherits the same property. 
Notably, the proofs in these papers often used different techniques for each protocol; our reduction lets either one suffice for both.

As an additional simplification, we can also notice that both of the results in \cite{escola2024quantum} now follow immediately from earlier literature combined with our reductions. 
In particular, in \cite{unruh2014quantum} security was shown for the $f$-measure protocol in the \textit{random oracle model}: the success probability for an $f$-measure protocol is small unless many oracle calls to $f$ are made by Alice and Bob.
Separately, security for $f$-route in the random oracle model was shown in \cite{escola2024quantum} using different techniques. 
Now, only one of these two proofs is needed, along with the observation that from our implication from $f$-route to $f$-measure, the oracle calls needed for $f$-route upper bound those for $f$-measure. 
In \cite{escola2024quantum} the authors also show strong parallel repetition of $f$-route in the no pre-shared entanglement model. 
Again, this property was known for $f$-measure \cite{tomamichel2013monogamy}, and is also carried over by our reductions. 

Beyond simplifying the literature, we also obtain a number of new properties of $f$-route and $f$-measure via our reductions.
Regarding lower bounds, \cite{asadi2024conditional} proved lower bounds on CDQS (equivalently, $f$-route) on communication and the dimension of the shared resource state in terms of various measures of communication complexity. 
For instance, among other bounds, they proved a lower bound on communication in CDQS for a function $f$ in terms of the communication cost of two-round quantum interactive proof for the same function. 
An alternative lower bound in terms of a public coin quantum interactive proof, again with two rounds but now with two provers, was proven in \cite{girish2025comparing}. 
Both of these are now inherited by $f$-measure.

Regarding upper bounds, more was previously known about $f$-route than about $f$-measure.
In particular \cite{cree2023code} gave an upper bound from span program size for $f$-route; $f$-measure now inherits this bound.
Via the connection to classical CDS, \cite{allerstorfer2024relating} gave upper bounds on $f$-route for all functions of $2^{O(\sqrt{n\log n})}$, which is also now inherited by $f$-measure.\footnote{Note that as of this writing, all classically controlled function NLQCs use sub-exponential entanglement even in the worst case. It would be interesting to understand if this is always the case.}
Finally \cite{allerstorfer2024relating} gave an efficient $f$-route protocol for quadratic residuosity, which we now obtain for $f$-measure.  

On a more negative side, proving linear lower bounds on CDQS or $f$-route was known to face a `CDS barrier': placing such a lower bound on entanglement in CDQS also puts the same lower bound on randomness in classical CDS, which is a long standing open problem in the classical literature. 
Now, we see that placing such lower bounds on $f$-measure also faces the same barrier. 

%%%%%%%%%%%%%%%%%%%%%%%%%%%%%%%%%%%%%%%%%%%%%%%%%%%%%%%%%%%%%%%%%%%%%%
\section{Discussion}
%%%%%%%%%%%%%%%%%%%%%%%%%%%%%%%%%%%%%%%%%%%%%%%%%%%%%%%%%%%%%%%%%%%%%%

Taking a comparative approach opens many new directions towards understanding NLQC and the subjects it is related to. 
Perhaps most clearly, future explorations of further NLQC classes and their relationships may yield more insight into entanglement cost in NLQC. 
For instance, inspired by our observation here that many superficially distinct classically controlled NLQC classes are equivalent, it is natural to ask if \emph{all} NLQC classes involving controlling an $O(1)$ size, non-trivial, quantum operation off of classical inputs could be equivalent.

A direction we have not explored here is to prove that some instances of NLQC \emph{cannot} be reduced to one another. 
For instance, it seems plausible that coherent protocols cannot be reduced to classically controlled ones. 
In fact, $Cf$-PHASE and hence $Cf$-SWAP are subject to the conditional exponential lower bounds from \cite{junge2022geometry}, while CDQS has a sub-exponential upper bound \cite{allerstorfer2024relating}. 
Assuming the conjectures needed for the exponential lower bound hold then, there can be no efficient implication from CDQS to $Cf$-SWAP. 
More generally, one could look for general properties that imply one NLQC cannot be reduced to another. 
Doing so would identify those properties as key to the ``hardness'' of a NLQC. 

One of the key open questions in NLQC is to prove a linear lower bound on a classically controlled function class. 
For instance, to prove a linear bound against $f$-route while allowing correctness errors. 
This would be highly desirable from the perspective of QPV. 
At the same time, such a bound would also imply a highly sought after lower bound for CDS, where a key goal is to show a linear lower bound on communication in the setting with correctness and privacy errors allowed.  
Our approach suggests a starting point to proving such lower bounds: pick a harder but related NLQC and try to lower bound it first. 
For instance, since $Cf$-PAULI implies $f$-route, placing a linear, robust, lower bound on $Cf$-PAULI is an easier but related problem to proving a robust lower bound on $f$-route and CDS. 

We can also hope that new reductions will lead to new connections among the many areas tied to NLQC. 
For instance, information-theoretic cryptography and in particular CDS is related to $f$-route, while $f$-measure seems to be related to uncloneable cryptography. 
In particular parallel repetition for $f$-measure is proven using the same key results on monogamy of entanglement games as appear in uncloneable cryptography \cite{tomamichel2013monogamy}. 
It would be interesting to further explore this connection, or other ways in which bridges among NLQCs can relate apparently distinct areas in quantum cryptography and information theory. 

Finally, as a speculative comment, we can ask if the connections between NLQC and complexity theory highlighted here can give insight into complexity theory, for instance by providing a route to new lower bounds on complexity measures. 
It is interesting to note here a comparison to communication complexity: in communication complexity it is possible to obtain good lower bounds on the needed non-local resource (communication), but the applications of this to complexity theory are restricted by the fact that communication never needs to be larger than linear. 
In NLQC, we again have some non-local resource we can hope to bound (correlation or entanglement), but now we can hope for super-linear lower bounds. 
Since entanglement in many NLQC classes is upper bounded by complexity measures, entanglement lower bounds in NLQC provide complexity lower bounds. 
This was already observed in earlier works, for instance in the context of $f$-routing \cite{buhrman2013garden, cree2023code}, but we hope that finding new complexity based upper bounds on further NLQC classes and exploring the relationships among classes can lead to interesting complexity lower bounds.  

%%%%%%%%%%%%%%%%%%%%%%%%%%%%%%%%%%%%%%%%%%%%%%%%%%%%%%%%%%%%%%%%%%%%%%
\subsection*{Acknowledgements}
%%%%%%%%%%%%%%%%%%%%%%%%%%%%%%%%%%%%%%%%%%%%%%%%%%%%%%%%%%%%%%%%%%%%%%

AB was supported by the ANR project PraQPV, grant number ANR-24-CE47-3023. AM and MS acknowledge the support of the Natural Sciences and Engineering Research Council of Canada (NSERC); this work was supported by an NSERC Discovery grant (RGPIN-2025-03966) and NSERC-UKRI Alliance grant (ALLRP 597823-24). 
Research at the Perimeter Institute is supported by the Government of Canada through the Department of Innovation, Science and Industry Canada and by the Province of Ontario through the Ministry of Colleges and Universities. PVL is supported by France 2030 under the French National Research Agency award number ANR-22-PETQ-0007. HY is supported by AFOSR award FA9550-23-1-0363, NSF CAREER award CCF-2144219, and the Sloan Foundation.
SH acknowledges that:
This work was supported by the MSCA cofund QuanG [grant number : 101081458]
\includegraphics[width=0.15\linewidth]{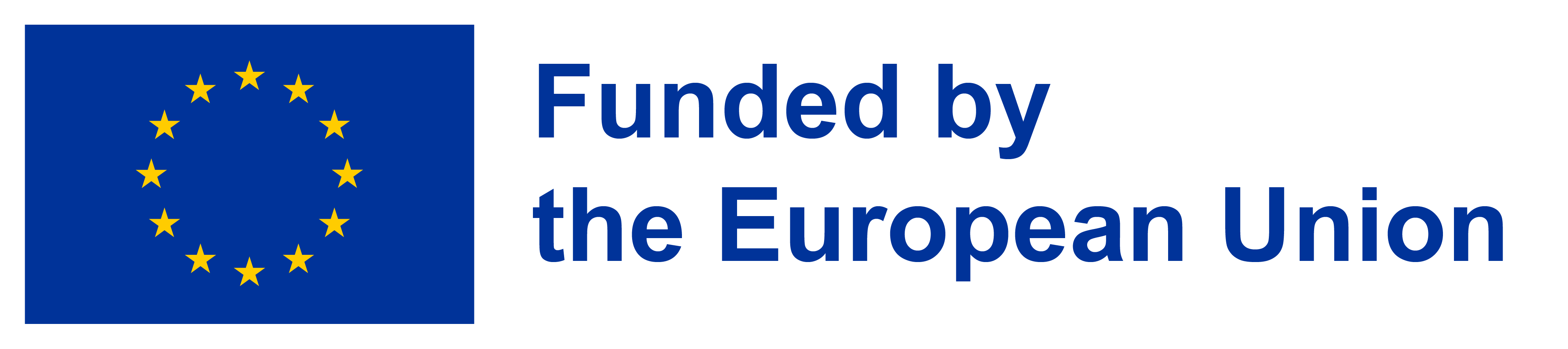}
\appendix

%%%%%%%%%%%%%%%%%%%%%%%%%%%%%%%%%%%%%%%%%%%%%%%%%%%%%%%%%%%%%%%%%%%%%%
\section{Properties of CDQS}\label{sec:CDQSproperties}
%%%%%%%%%%%%%%%%%%%%%%%%%%%%%%%%%%%%%%%%%%%%%%%%%%%%%%%%%%%%%%%%%%%%%%

\begin{theorem}\label{thm:CDQStoCDQS'}
    A CDQS protocol using resource system $\Psi$ implies a CDQS' protocol for the same function using the same resource system.
\end{theorem}

\begin{proof}
    \noindent Take the classical secret $s\in S$ for the CDQS' protocol and record it as a computational basis state $\ket{s}_Q$, which serves as the secret for the CDQS protocol. 
    The referee runs the decoder for the CDQS protocol and then measures the outcome in the computational basis, and takes the output as his guess for the value of $S$.

    \vspace{0.2cm}
    \noindent \textbf{Correctness:} We have that $f(x,y)=1$. By correctness of the CDQS protocol we have
    \begin{align}
        \left\Vert \mathbfcal{D}^{x,y}_{M\rightarrow Q}\circ \mathbfcal{N}^{x,y}_{Q\rightarrow M} - \mathbfcal{I}_Q \right\Vert_\diamond \leq \epsilon.
    \end{align}
    Evaluating these channels on the input state $\ket{s}_Q$, we obtain
    \begin{align}\label{eq:TDbound}
        \left\Vert \mathbfcal{D}^{x,y}_{M\rightarrow Q}\circ \mathbfcal{N}^{x,y}_{Q\rightarrow M}(\ketbra{s}{s}_Q) - \ketbra{s}{s}_Q \right\Vert_1 \leq \epsilon\,.
    \end{align}
    To obtain a classical outcome for the CDQS protocol, the referee then measures the output of the decoder in the computational basis. 
    Label the outcome of this measurement as $s'$. 
    Their probability of a correct outcome is then
    \begin{align}
        \text{Pr}[s'=s] = \bra{s}\mathbfcal{D}^{x,y}_{M\rightarrow Q}\circ \mathbfcal{N}^{x,y}_{Q\rightarrow M}(\ketbra{s}{s}_Q)\ket{s} = F(\ketbra{s}{s}, \mathbfcal{D}^{x,y}_{M\rightarrow Q}\circ \mathbfcal{N}^{x,y}_{Q\rightarrow M}(\ketbra{s}{s}_Q)).
    \end{align}
    Using the Fuchs Van de Graaf inequalities \eqref{eq:FVDG} and the bound \eqref{eq:TDbound} on the trace distance, we obtain
    \begin{align}
        \text{Pr}[s'=s] &= F(\ketbra{s}{s}, \mathbfcal{D}^{x,y}_{M\rightarrow Q}\circ \mathbfcal{N}^{x,y}_{Q\rightarrow M}(\ketbra{s}{s}_Q)) \nonumber \\
        &\geq \left(1 - \epsilon/2\right)^2,
    \end{align}
    so that we obtain a perfectly correct CDQS' protocol as the CDQS protocol becomes perfectly correct, as needed. 

    \vspace{0.2cm}
    \noindent \textbf{Security:} We have that $f(x,y)=0$. From the definition of security in the CDQS protocol, we have that there exists a simulator channel $\mathbfcal{S}_{\varnothing \rightarrow M}^{x,y}$ such that
    \begin{align}
        \|\mathbfcal{S}_{\varnothing \rightarrow M}^{x,y} \circ \tr_Q - \mathbfcal{N}_{Q\rightarrow M}^{x,y}\|_\diamond \leq \delta.
    \end{align}
    Consider taking as input the state $\ket{s}_Q$, as we do in the CDQS' protocol. 
    This gives
    \begin{align}
        \left\Vert \sigma_M(x,y) - \rho_M(s,x,y)\right\Vert_1 \leq \delta,
    \end{align}
    where $\sigma_M(x,y)=\mathbfcal{S}_{\varnothing \rightarrow M}^{x,y} \circ \tr_Q(\ketbra{s}{s}_Q)$. 
    Taking these $\sigma_M(x,y)$ as our simulator distribution, this is exactly $\delta$ security of the CDQS' protocol. 
\end{proof}

\begin{theorem}\label{thm:CDQS'toCDQS}
    A CDQS' protocol using resource system $\Psi$ implies a CDQS protocol for the same function using resource system $\Psi^{\otimes 2}$.
\end{theorem}

\begin{proof}
    \noindent Alice takes the quantum system $Q$ input to the CDQS protocol and applies the one-time pad, using $2 \log d_Q$ bits of classical key, which we call $s$. 
    We use two instances of the CDQS' protocol, each of which hide $\log d_Q$ bits of secret, to hide the $2\log d_Q$ bits of key for the one-time pad. 
    A lemma in \cite{allerstorfer2024relating} shows the correctness and security errors of the CDQS' protocol are at most twice that of the original CDQS' protocol. 
    The encoded system $Q$ is sent to the referee, along with the message systems for the CDQS' protocol. 
    The channel applied by Alice and Bob's combined actions is then
    \begin{align}\label{eq:CDQSchannel}
        \mathbfcal{N}^{x,y}_{Q\rightarrow QM}(\cdot) = \frac{1}{2^{|s|}} \sum_{m,s} P_Q^s(\cdot)P_Q^s \otimes \rho_M(s,x,y),
    \end{align}
    where $M$ is the message system for (both instances of) the CDQS' protocol. 

    \vspace{0.2cm}
    \noindent \textbf{Correctness:} We have that $f(x,y)=1$. Correctness of the CDQS' states that there exists a decoder $\mathbfcal{D}^{'x,y}_{M\rightarrow S}$ such that
    \begin{align}\label{eq:foralls}
        \forall s\in S,\,\,\, \text{Pr} \left[\mathbfcal{D}^{'x,y}_{M\rightarrow S}(\rho_M(x,y,s)) =s\right] \geq 1-2\epsilon.
    \end{align}
    The $2\epsilon$ rather than $\epsilon$ is because we need to use two instances of the CDQS' protocol, as mentioned.
    We can also frame this condition in terms of the conditional probabilities, 
    \begin{align}
        p_{s'|s} = \text{Pr} \left[s'=\mathbfcal{D}^{'x,y}_{M\rightarrow S}(\rho_M(x,y,s))\right].
    \end{align}
    In particular \Cref{eq:foralls} is the statement that $p_{s'=s|s}\geq 1-2\epsilon$ for all $s$.
    This also means that $\sum_{s'\neq s}p_{s'|s}\leq 2\epsilon$, so that 
    \begin{align}\label{eq:4epsilon}
        \frac{1}{2^{|s|}}\sum_{s',s} |\delta_{s's}-p_{s'|s}| = \frac{1}{2^{|s|}}\sum_s\left(\sum_{s'=s}|1-p_{s|s}| + \sum_{s'\neq s}p_{s'|s}\right) \leq \frac{1}{2^{|s|}} \sum_s 4\epsilon = 4\epsilon.
    \end{align}
    To recover the secret, the referee applies $\mathbfcal{D}^{'x,y}_{M\rightarrow S}$ to the message system he receives, producing a variable $s'$, and then applies $P^{s'}_Q$, then finally traces out $M$. 
    Let the combined action of these steps be denoted by the channel $\mathbfcal{D}^{x,y}_{M\rightarrow S}$.
    This produces
    \begin{align}
        \mathbfcal{D}^{x,y}_{QM\rightarrow Q}\circ \mathbfcal{N}^{x,y}_{Q\rightarrow QM}(\cdot) = \frac{1}{2^{|s|}} \sum_{s',s}  p_{s'|s} P_Q^{s-s'}(\cdot)P_Q^{s-s'} .
    \end{align}
    We need to show this is close to the identity channel, which we do in the following calculation 
    \begin{align}
        \left\Vert \frac{1}{2^{|s|}}\sum_{s',s}  p_{s'|s} P_Q^{s-s'}(\cdot)P_Q^{s-s'} - \mathbfcal{I}_Q\right\Vert_\diamond &= \left\Vert\frac{1}{2^{|s|}}\sum_{s',s}  p_{s'|s} P_Q^{s-s'}(\cdot)P_Q^{s-s'} - \frac{1}{2^{|s|}}\sum_{s',s}  \delta_{s's} P_Q^{s-s'}(\cdot)P_Q^{s-s'}\right\Vert_\diamond \nonumber \\
        &= \sup_d \max_{\Psi_{R_dQ}} \frac{1}{2^{|s|}}\left\Vert\sum_{s',s}  \left(p_{s'|s} P_Q^{s-s'}\Psi_{R_dQ}P_Q^{s-s'} -   \delta_{s's} P_Q^{s-s'}\Psi_{R_dQ}P_Q^{s-s'}\right)\right\Vert_1 \nonumber \\
        &\leq  \sup_d \max_{\Psi_{R_dQ}} \frac{1}{2^{|s|}}\sum_{s',s}\left\Vert  \left(p_{s'|s}  -   \delta_{s's} \right)P_Q^{s-s'}\Psi_{R_dQ}P_Q^{s-s'}\right\Vert_1 \nonumber \\
        &\leq  \left(\frac{1}{2^{|s|}}\sum_{s',s}  |p_{s'|s}-\delta_{s's}| \right) \sup_d \max_{\Psi_{R_dQ},s,s'}  \left\Vert P_Q^{s-s'} \Psi_{R_dQ} P_Q^{s'-s}\right\Vert_1\nonumber \\
        &=4\epsilon.
    \end{align}
    We used sub-additivity in the first inequality. The second inequality follows by bringing the $p_{s'|s}  -   \delta_{s's}$ factor out using absolute homogeneity, and then moving the sup and max through the sum by adding the maximization over $s,s'$.
    Finally we use \Cref{eq:4epsilon} to obtain the last line.  
    This is $4\epsilon$-correctness of the CDQS protocol.
    
    \vspace{0.2cm}
    \noindent \textbf{Security:} We have now $f(x,y)=0$. We choose the simulator channel
    \begin{align}
    \mathbfcal{S}_{\varnothing \rightarrow MQ}^{x,y} = \frac{\mathcal{I}_Q}{d_Q}\otimes \sigma_M(x,y),
    \end{align}
    where $\sigma_M(x,y)$ is the family of channels appearing in the security definition for the CDQS' protocol, which recall states that
    \begin{align}\label{eq:security'}
        \forall s\in S,\,\,\, \left\Vert \sigma_M(x,y) - \rho_M(s,x,y) \right\Vert_1 \leq \delta.
    \end{align}
    We need to bound the following diamond norm, 
    \begin{align}
    \Vert \mathbfcal{S}_{\varnothing \rightarrow MQ}^{xy}&\circ \tr_Q - \mathbfcal{N}_{Q\rightarrow QM}\Vert_\diamond \nonumber \\
    &= \sup_{n}\max_{\Psi_{R_nQ}}\left\Vert \mathbfcal{S}_{\varnothing \rightarrow MQ}^{xy}\circ \tr_Q(\Psi_{R_nQ}) - \mathbfcal{N}^{xy}_{Q\rightarrow QM}(\Psi_{R_nQ})\right\Vert_1 \nonumber \\
    &=\sup_{n}\max_{\Psi_{R_nQ}}\left\Vert\Psi_{R_n}\otimes \frac{\mathcal{I}_Q}{d_Q}\otimes \sigma_M(x,y) - \frac{1}{2^{|s|}}\sum_{s}P^s_Q\Psi_{R_nQ}P^s_Q\otimes \rho_M(s,x,y)\right\Vert_1 \nonumber \\
    &= \sup_{n}\max_{\Psi_{R_nQ}}\left\Vert\frac{1}{2^{|s|}}\sum_{s}P^s_Q\Psi_{R_nQ}P^s_Q\otimes \sigma_M(x,y)-\frac{1}{2^{|s|}}\sum_{s}P^s_Q\Psi_{R_nQ}P^s_Q\otimes \rho_M(s,x,y)\right\Vert_1 \nonumber \\
    &\leq \left\Vert\sigma_M(x,y)-\rho_M(s,x,y)\right\Vert_1 \,\sup_{n}\max_{\Psi_{R_nQ}}\left\Vert \frac{1}{2^{|s|}}\sum_{s}P^s_Q\Psi_{R_nQ}P^s_Q \right\Vert_1 \nonumber \\
    &\leq 2\delta.
\end{align}
We used sub-multiplicativity of the one-norm in the first inequality, and the security statement \eqref{eq:security'} along with normalization of the state $\Psi_{R_dQ}$ in the second inequality. 
The above is exactly $2\delta$ security of the CDQS protocol.\footnote{Note $2\delta$ is appearing again rather than $\delta$ because we use two instances of the CDQS' protocol.}
\end{proof}

\begin{theorem}\label{thm:CDQSprimearbitrary}
    Suppose we have a CDQS' protocol which hides a uniformly random classical secret $r$. Then, at the cost of $|r|$ bits of communication, we can build a protocol that hides a secret $s$ with $|s|=|r|$ with an arbitrary distribution. 
\end{theorem}

\begin{proof}
    Run the CDQS' protocol on secret $r$, and additionally have Alice send $s'=s\oplus r$. 

    \vspace{0.2cm}
    \noindent \textbf{Correctness:} To decode, the referee recovers $r'$ from the CDQS' protocol, which will equal the input $r$ with probability at least $1-\epsilon$. They then compute $s'\oplus r'$ as output. Since $s'=s\oplus r$, this will be correct whenever $r'=r$, so with probability at least $1-\epsilon$. 
    Thus the protocol storing $s$ is $1-\epsilon$ correct. 

    \vspace{0.2cm}
    \noindent \textbf{Security:} Security of the CDQS' protocol storing $r$ implies that
    \begin{align}
        \forall (x,y)\in X\times Y \,\,\, s.t. \,\, f(x,y)=0,\,\,\,\forall s\in S,\,\,\, \left\Vert \sigma_M(x,y) - \rho_M(s,x,y) \right\Vert_1 \leq \delta.
    \end{align}
    For the CDQS' protocol storing $s$, we take 
    \begin{align}
        \tilde{\sigma}_{SM}(x,y) = \pi_{A'}\otimes \sigma_M(x,y) = \frac{1}{2^{|r|}}\sum_{r} \ketbra{s\oplus r}{s\oplus r} \otimes \sigma_M(x,y),
    \end{align}
    where $\pi_{A'}$ is the maximally mixed state on $S$.
    Then by the triangle inequality,
    \begin{align}
        \left\Vert\tilde{\sigma}_{SM}(x,y) - \frac{1}{2^{|r|}}\sum_{r} \ketbra{s\oplus r}{s\oplus r}_{A'}\otimes \rho_{M}(r,x,y) \right\Vert_1 &= \left\Vert\frac{1}{2^{|r|}}\sum_{r}\ketbra{s\oplus r}{s\oplus r}_{A'} \otimes \sigma_M(x,y) \right. \nonumber \\ 
        &\qquad \quad  - \left. \frac{1}{2^{|r|}}\sum_{r} \ketbra{s\oplus r}{s\oplus r}_{A'}\otimes \rho_{M}(r,x,y) \right\Vert_1 \nonumber \\
        &= \left\Vert \sigma_M(x,y) - \rho_{M}(r,x,y) \right\Vert_1 \nonumber \\
        &\leq \delta \nonumber 
    \end{align}
    so that the protocol hiding $s$ is also $\delta$ secure. 
\end{proof}

%%%%%%%%%%%%%%%%%%%%%%%%%%%%%%%%%%%%%%%%%%%%%%%%%%%%%%%%%%%%%%%%%%%%%%
\section{Properties of \texorpdfstring{$Cf$}{TEXT}-PHASE}\label{sec:phaseproperties}
%%%%%%%%%%%%%%%%%%%%%%%%%%%%%%%%%%%%%%%%%%%%%%%%%%%%%%%%%%%%%%%%%%%%%%

\reoplusphase*

\vspace{0.2cm}
\begin{proof}
    The protocol proceeds by 
    \begin{enumerate}
        \item In the first round operations, copying the inputs in the computational basis.
        \item Executing the protocol for $f_1$ on the first copy and the protocol for $f_2$ on the second
        \item Inverting the copying procedure to reset the ancillas in the second round operations. 
    \end{enumerate} 
    Following these steps, let us first evaluate the action of the protocol in the case where $\epsilon_1=\epsilon_2=0$ on a basis element $\ket{x}_A\ket{y}_B$,
    \begin{align}
        \ket{x}_A\ket{y}_B \ket{0}_{\bar{A}}\ket{0}_{\bar{B}}&\overset{1)}{\rightarrow} \ket{x}_A\ket{y}_B\ket{x}_{\bar{A}} \ket{y}_{\bar{B}} \nonumber \\
        &\overset{2)}{\rightarrow} ((-1)^{f_1(x,y)}\ket{x}_A\ket{y}_B)((-1)^{f_2(x,y)}\ket{x}_{\bar{A}} \ket{y}_{\bar{B}}) \nonumber \\
        &=(-1)^{f_1(x,y)\oplus f_2(x,y)} \ket{x}_{A} \ket{y}_{B}\ket{x}_{\bar{A}} \ket{y}_{\bar{B}} \nonumber \\
        &\overset{3)}{\rightarrow}(-1)^{f_1(x,y)\oplus f_2(x,y)} \ket{x}_{A} \ket{y}_{B}\ket{0}_{\bar{A}} \ket{0}_{\bar{B}}\,.
    \end{align}
    We see that this is an exactly correct protocol for $f_1\oplus f_2$ whenever the protocols for $f_1$ and $f_2$ are exactly correct. 
    
    It is clear that the reduction is an oracle reduction, so \Cref{remark:blackboxerror} leads to the additive error property. 
\end{proof}

Next consider \Cref{thm:PhaseANDclosure}, which deals with taking the AND of $f(x,y)$ with a single bit $z$.

\rePhaseANDclosure*

\vspace{0.2cm}
\begin{proof}
Consider the function 
\begin{align}
    f(x,y\wedge Z)\equiv f(x_1,...,x_n,y_1\wedge z, ...,y_n\wedge z).
\end{align}
This can be implemented by having Bob, who holds $\ket{y}$ and $\ket{z}$, compute the AND locally and input $\ket{y\wedge Z}$ into the phase protocol for $f(x,y)$. 
Then when $z=1$, this implements the original function $f(x,y)$, while if $z=0$ this implements $f(x,0)$. 
This is not quite performing $f(x,y)\wedge z$, but is performing something close: it only differs in $z=0$ cases, and in that case differs by a function computed only from $x$ (in particular $f(x,0)$ is applied, rather than applying no phase). 

To correct this, it suffices to have Alice and Bob implement using \Cref{thm:oplusphase} the two functions $f_1(x,y,z)=f(x,0)\wedge (z\oplus 1)$, and $f_2(x,y,z)=f(x,y\wedge Z)$. 
Then using that 
\begin{align}
    f(x,y)\wedge z=[f(x,0)\wedge (z\oplus 1)]\oplus[f(x,y\wedge Z)]
\end{align}
this will correctly implement $f(x,y)\wedge z$ as needed. 
We noted above how to implement $f(x,y\wedge Z)$ using a single instance of the protocol for $f(x,y)$. 
It remains to see that we can implement $f(x,0)\wedge (z\oplus 1)$ with a single EPR pair $\Psi^+$. 

To do this, Alice locally computes $f(x,0)$ coherently from $\ket{x}$ and stores it in a register $A_0$. 
Meanwhile, Bob computes $z\oplus 1$ coherently and stores it in a register $B_0$. 
Then, they execute a coherent phase protocol on $A_0$ and $B_0$ for the AND function; this amounts to implementing a $CZ_{A_0B_0}$ gate which can be done with one EPR pair by a standard teleportation strategy. 
Alice and Bob then uncompute $f(x,0)$ and $z\oplus 1$ in their second round operations to reset the $A_0$ and $B_0$ registers to $0$. 

Using \Cref{remark:blackboxerror} and the fact that we use a single invocation of the protocol for $f(x,y)$ (along with an exact protocol for a single AND function) the overall error is the same as the error in the protocol for $f(x,y)$. 
\end{proof}

Next, we will develop the proof of the more general \Cref{thm:phaseclosure}, which replaces $z$ with an arbitrary function. 
We restate the theorem before giving the proof. 

\rephaseclosure*

\vspace{0.2cm}
\begin{proof}
    Write $g(k,\ell)$ in its rank decomposition, 
    \begin{align}
        f(x,y)\wedge g(k,\ell) &= f(x,y) \wedge \bigoplus_{i=1}^m h_{i}(k)\wedge h_i'(\ell) \nonumber \\
        &= \bigoplus_{i=1}^m f(x,y)\wedge h_i(k) \wedge h_i'(\ell).
    \end{align}
    In the second line we distributed the first AND operator over the sum. 
    Next, use that from \Cref{thm:oplusphase} we can implement the XOR of a collection of functions by implementing each function in parallel, so it suffices to implement each of the functions $f(x,y)\wedge h_i(k) \wedge h_i'(\ell)$. 
    To do this, we apply \Cref{thm:PhaseANDclosure} twice: by one invocation of the theorem a protocol for $f(x,y)$ allows us to implement $f(x,y)\wedge h_i(k)$ by adding one EPR pair to the resource system, and then by a second invocation the resulting protocol for $f(x,y)\wedge h_i(k)$ plus one further EPR pair gives an implementation of $f(x,y)\wedge h_i(k) \wedge h_i'(\ell)$. 
    Thus, if $\Psi$ is the resource system for $f(x,y)$ we can implement $f(x,y)\wedge h_i(k)\wedge h(\ell)$ using $\Psi\otimes (\Psi^+)^{\otimes 2}$.
    Repeating this for each of the $m$ terms in the above decomposition, we find that $f(x,y)\wedge g(k,\ell)$ can be implemented using the resource system $\Psi^{\otimes m}\otimes (\Psi^+)^{\otimes 2m}$.

    Noting that we use the protocol for $f(x,y)$ a total of $m$ times and using \Cref{remark:blackboxerror} we have that the overall error is $\epsilon \cdot m$, where $\epsilon$ is the error in the coherent phase protocol for $f(x,y)$. 
\end{proof}

%%%%%%%%%%%%%%%%%%%%%%%%%%%%%%%%%%%%%%%%%%%%%%%%%%%%%%%%%%%%%%%
\section{Properties of \texorpdfstring{$Cf$-$\mathbf{U}$}{TEXT}}\label{sec:coherentUproperties}
%%%%%%%%%%%%%%%%%%%%%%%%%%%%%%%%%%%%%%%%%%%%%%%%%%%%%%%%%%%%%%%

We record the proofs of the composition properties given for the coherent unitary NLQC.

\retensorproduct*

\begin{proof}
The protocol for the tensor product unitary is
    \begin{enumerate}
    \item Initialize the ancilla systems $\bar{A}$ and $\bar{B}$ on the respective sides in state $\ket{0}$.
    \item Apply CNOTs to $A\bar{A}$ and $B\bar{B}$ locally to copy $A$ and $B$ in the computational basis.
    \item Implement coherent $\mathbf{U}_1$ on $A'AB$ and $\mathbf{U}_2$ on $B'\bar{A}\bar{B}$.
    \item Locally uncompute $\bar{A}$ and $\bar{B}$ using again the same CNOTs.
\end{enumerate}
To see that this succeeds explicitly, we can track the state through the steps in the protocol,
\begin{align}
    \ket{\Psi}_{A'B'AB}\otimes\ket{0}_{\bar{A}}\ket{0}_{\bar{B}} &= \underset{abxy}{\sum}c_{abxy}\ket{a}_{A'}\ket{b}_{B'}\ket{x}_A\ket{y}_B\otimes \ket{0}_{\bar{A}}\ket{0}_{\bar{B}} \\
    \overset{CNOT_{A\bar{A}}\otimes CNOT_{B\bar{B}}}{\rightarrow} \hskip 1cm &\underset{abxy}{\sum}c_{abxy}\ket{a}_{A'}\ket{b}_{B'}\ket{x}_A\ket{y}_B\otimes \ket{x}_{\bar{A}}\ket{y}_{\bar{B}} \nonumber \\
    \overset{(\mathbf{U}^f_1)_{A'AB}\otimes (\mathbf{U}^f_2)_{B'\bar{A}\bar{B}}}{\rightarrow} \hskip 1cm &\underset{abxy} {\sum}c_{abxy}\,\mathbf{U}^{f(x,y)}_1\ket{a}_{A'}\mathbf{U}^{f(x,y)}_2\ket{b}_{B'}\ket{x}_A\ket{y}_B\otimes \ket{x}_{\bar{A}}\ket{y}_{\bar{B}} \nonumber \\
    \overset{CNOT_{A\bar{A}}\otimes CNOT_{B\bar{B}}}{\rightarrow} \hskip 1cm &\underset{abxy}{\sum}c_{abxy}\,\mathbf{U}^{f(x,y)}_1\ket{a}_{A'}\mathbf{U}^{f(x,y)}_2\ket{b}_{B'}\ket{x}_A\ket{y}_B\otimes \ket{0}_{\bar{A}}\ket{0}_{\bar{B}}.
\end{align}
Tracing out the ancillas yields the desired output. 

Notice that the construction is an oracle reduction, so has an additive error by \Cref{remark:blackboxerror}.
\end{proof}

\repowers*

\vspace{0.2cm}
\begin{proof}
    Since $\mathbf{U}^1_{A'}$ and $\mathbf{U}^2_{A'}$ commute, they must share an eigenbasis, which we label $\{\ket{u_k}\}_k$.
    Label the eigenvalues of $\mathbf{U}^1_{A'}$ as $\{\mu_i\}_i$ and the eigenvalues of $\mathbf{U}^2_{A'}$ as $\{\nu_i\}_i$.
    Alice and Bob execute the following protocol. 
    \begin{enumerate}
        \item Alice copies system $A'$ in the eigenbasis $\ket{u_k}$, producing systems $A'_1$ and $A'_2$.
        \item Alice and Bob run a NLQC protocol for $\mathbf{U}_{A'_1}^1\otimes \mathbf{U}^2_{A'_2}$.
        \item Alice undoes the copying operation, leaving only output system $A'_1=A'$. 
    \end{enumerate}
    To see the effect of this protocol, consider an arbitrary basis state $\ket{u_i}_{A'}\ket{x}_A\ket{y}_B$. 
    Then the steps above produce
    \begin{align}
        \ket{u_i}_{A'}\ket{x}_A\ket{y}_B &\overset{1}{\rightarrow} \ket{u_i}_{A'_1}\ket{u_i}_{A'_2}\ket{x}_A\ket{y}_B \nonumber \\
        & \overset{2}{\rightarrow} (\mu_i \ket{u_i}_{A'_1})(\nu_i\ket{u_i}_{A'_2})\ket{x}_A\ket{y}_B \nonumber \\
        &\overset{3}{\rightarrow} (\mu_i\nu_i)\ket{u_i}_{A'}\ket{x}_A\ket{y}_B \nonumber \\
        &= \mathbf{U}^1_{A'}\mathbf{U}^2_{A'}\ket{u_i}_{A'}\ket{x}_A\ket{y}_B\,.
    \end{align}
    Since this holds for a basis, it also holds for an arbitrary input state. 
    Notice that the construction is an oracle reduction, so has an additive error by \Cref{remark:blackboxerror}.
\end{proof} 

\bibliographystyle{unsrtnat}
\bibliography{biblio}

\end{document}